\documentclass[letterpaper,11pt]{article}

\usepackage{amsmath,amssymb,amsthm}
\usepackage{authblk} 
\usepackage[american]{babel}
\usepackage{balance} 
\usepackage{enumitem} 
\usepackage{float} 
\usepackage[T1]{fontenc}
\usepackage{fullpage} 
\usepackage{graphicx}
\usepackage[utf8]{inputenc}
\usepackage{layout}
\usepackage{mathtools} 
\usepackage{mdframed}
\usepackage{diagbox}
\usepackage{multirow} 
\usepackage{subcaption}
\usepackage{tikz}
\usepackage{titling}
\usepackage{thmtools}
\usepackage{url}
\usepackage[normalem]{ulem} 
\usepackage{xcolor}
\usepackage[backref=page, pagebackref=true, linktocpage, hypertexnames=false]{hyperref} 
\usepackage{algorithm}
\usepackage[noend]{algpseudocode}
\usepackage{cleveref} 
\usepackage{comment}
\usepackage{bm} 

\definecolor{redlink}{rgb}{0.6, 0, 0}
\definecolor{greenlink}{rgb}{0, 0.6, 0}
\definecolor{bluelink}{rgb}{0, 0, 0.6}
\hypersetup{
	colorlinks=true,
	linkcolor=redlink, 
	urlcolor=redlink, 
	citecolor=bluelink, 
}

\newcommand{\real}{\mathbb{R}}

\graphicspath{{figures/}}
\theoremstyle{definition}
\newtheorem{definition}{Definition}[section]
\newtheorem{question}{Question}[]
\theoremstyle{plain}
\newtheorem{theorem}{Theorem}[section]
\newtheorem{lemma}{Lemma}[section]
\newtheorem{invariant}{Invariant}
\newtheorem{corollary}{Corollary}[section]

\newtheorem{observation}{Observation}[section]

\theoremstyle{remark}
\newtheorem{remark}{Remark}[section]
\newtheorem{claim}{Claim}[section]

\crefname{claim}{claim}{claims}

\setenumerate{itemsep=0mm, topsep=0mm, parsep=0mm}

\makeatletter 
\newcounter{algorithmicH}
\let\oldalgorithmic\algorithmic
\renewcommand{\algorithmic}{%
  \stepcounter{algorithmicH}
  \oldalgorithmic}
\renewcommand{\theHALG@line}{ALG@line.\thealgorithmicH.\arabic{ALG@line}}
\makeatother
\algnewcommand\Break{\textbf{break}}
\algnewcommand\Continue{\textbf{continue}}
\algnewcommand\Exit{\textbf{exit}}
\algnewcommand\Or{\textbf{or~}}
\algnewcommand\And{\textbf{and~}}
\algnewcommand{\IfThenElse}[3]{\State \algorithmicif\ #1\ \algorithmicthen\ #2\ \algorithmicelse\ #3}
\algnewcommand{\IfThen}[2]{\State \algorithmicif\ #1\ \algorithmicthen\ #2}
\algnewcommand{\ForInline}[2]{\State \algorithmicfor\ #1\ \algorithmicdo\ #2}

\newcommand{\eps}{\epsilon} 
\newcommand{\ignore}[1]{} 

\def\EMPH#1{\emph{\textcolor{redlink} {#1}}}

\DeclareMathOperator{\diam}{\mathrm{diam}}
\DeclareMathOperator{\port}{\mathrm{port}}
\DeclareMathOperator{\pos}{\mathrm{pos}}
\DeclareMathOperator{\code}{\mathrm{code}}
\DeclareMathOperator{\bin}{\mathrm{bin}}
\DeclareMathOperator{\lbl}{\mathrm{label}}

\DeclareMathOperator*{\level}{level}

\DeclarePairedDelimiterX{\norm}[1]{\lVert}{\rVert}{#1}
\DeclarePairedDelimiter\abs{\lvert}{\rvert}
\DeclarePairedDelimiter\ceil{\lceil}{\rceil}

\title{Covering the Euclidean Plane by a Pair of Trees}
\date{}
\author{
		Hung Le\thanks{University of Massachusetts Amherst, \href{}{hungle@cs.umass.edu}}
        \and
        Lazar Milenkovi\'{c}\thanks{Tel Aviv University, \href{}{milenkovic.lazar@gmail.com}}\and	
		Shay Solomon\thanks{Tel Aviv University, \href{}{solo.shay@gmail.com}}\and	
		Tianyi Zhang\thanks{ETH Zürich, \href{}{tianyi.zhang@inf.ethz.ch}}
	}	
\begin{document}

\maketitle

\begin{abstract}
A \emph{$t$-stretch tree cover} of a metric space $M = (X,\delta)$, for a parameter $t \ge 1$, is a collection of trees such that every pair of points has a
$t$-stretch path in one of the trees. 
Tree covers provide an important sketching tool that has found various applications over the years.
The celebrated \emph{Dumbbell Theorem} by Arya et al.\ [STOC'95] states that any set of points in the Euclidean plane admits a $(1+\epsilon)$-stretch tree cover with $O_\epsilon(1)$ trees.
This result extends to any (constant) dimension and was also generalized for arbitrary doubling metrics by Bartal et al.\ [ICALP'19].

Although the number of trees provided by the Dumbbell Theorem is constant, this constant is not small, even for a stretch significantly larger than $1+\epsilon$. At the other extreme, any single tree on the vertices of a regular $n$-polygon must incur a stretch of $\Omega(n)$.
Using known results of ultrametric embeddings,
one can easily get a stretch of $\tilde{O}(\sqrt{n})$ using two trees.
The question of whether a low stretch can be achieved using two trees has remained illusive, even in the Euclidean plane.

In this work, we resolve this fundamental question 
in the affirmative by presenting a constant-stretch cover with a pair of trees, for any set of points in the Euclidean plane.
Our main technical contribution is a \textbf{surprisingly simple} Steiner construction, for which we provide a \textbf{tight} stretch analysis of $\sqrt{26}$. The Steiner points can be easily pruned if one is willing to increase the stretch by a small constant. Moreover, we can bound the maximum degree of the construction by a constant.

Our result thus provides a simple yet effective reduction tool---for problems that concern approximate distances---from the Euclidean plane to a pair of trees.
To demonstrate the potential power of this tool, we present some applications for routing algorithms, including a constant-stretch compact routing scheme when \emph{handshaking} is allowed, on top of a pair of trees, in which the total memory usage is just $(2 + o(1))\log n$ bits.

\end{abstract}

\newpage

\section{Introduction}
A  \EMPH{$t$-stretch distance sketch} of a metric space, for a parameter $t \ge 1$ that is called the \emph{stretch}, is a \emph{compact structure} that approximately preserves the distances between every pair of points up to a (multiplicative) factor of $t$. The most basic \emph{compactness measure} is the size (i.e., the number of edges) of the sketch. The compactness of a distance sketch makes it a powerful primitive for countless algorithmic tasks.

Perhaps the most basic distance sketch is a \emph{$t$-spanner}, a (sub)graph preserving all pairwise distances up to a factor of $t$. Spanners have found many applications over the years, for example, in wireless and sensor networks \cite{RW04,BDS04,SS10} and distributed computing~\cite{ABP90,ABP91,BEIK19,Robinson21}. Nonetheless, the structure of a spanner could be too complex for a wide range of applications. A \emph{tree cover}~\cite{AP92,GKR01} is an important {\em structural} distance sketch composed of (a few) trees. 

\begin{definition}
A \EMPH{$t$-tree cover} of a metric $(X,\|\cdot\|_X)$  
is a set $\mathcal{T}$ of (edge-weighted) trees such that $X\subseteq V(T')$ for every tree $T'\in \mathcal{T}$, and for every two points $x,y\in X$ there exists a tree $T \in \mathcal{T}$ such that:
\begin{equation}\label{eq:tree-cover-def}
    \|xy\|_X \leq d_T(x,y)\leq t\cdot  \|xy\|_X
\end{equation}
That is, there exists a tree in the cover preserving the distance $\|xy\|_X$ up to the \emph{stretch factor} $t$. 
We say that a tree cover of a metric is \EMPH{spanning} if every tree in $\mathcal{T}$ is a spanning tree of the metric, i.e., none of the trees in $\mathcal{T}$ use additional \emph{Steiner} points that are not part of the metric.
\end{definition}

Since trees are arguably the simplest type of graphs, they provide a very useful form of structural distance sketches, which opens up a much broader range of applications ~\cite{TZ01Oracle,TZ06,CCL+23FOCS,CCL+24SODA,CHJ20,KLMS22,CCLST25};
for example, tree covers gave rise to optimal distance oracles~\cite{CCL+23FOCS,CCL+24SODA} and routing schemes~\cite{KLMS22,CCLST25}.

The most basic goal is to construct a $t$-tree cover for a small $t$ with a small number of trees. The union of all the trees in a $t$-cover forms a $t$-spanner of $(X,\|\cdot\|_X)$. In contrast, a $t$-spanner does not necessarily give rise to a $t$-cover, meaning that the notion of a tree cover is a strictly stronger notion than a spanner. 
A long line of research has been devoted to understanding whether one can match the stretch-sparsity tradeoff attained by spanners via tree covers, i.e., achieve the compactness of spanners together with the structural simplicity of trees. 

It is known since the early 90s that every $n$-point set in $\real^d$ (for any constant $d$) admits a $(1+\eps)$-spanner with $O(\eps^{1-d}n)$ edges~\cite{Keil88,RS91}, 
and this is asymptotically tight: there {exists} a point set $S$ in $\real^d$ such that any $(1+\eps)$-spanner for $S$ must have  $\Omega(\eps^{1-d}n)$ edges \cite{LS19}. 
The celebrated \emph{Dumbbell Theorem} by Arya et al.\ \cite{ADMSS95} extends this fundamental result of spanners to tree covers, providing a $(1+\eps)$-tree cover with $O(\eps^{-d} \log(1/\eps))$ trees, and thus with a bound of $O(\eps^{-d} \log(1/\eps) n)$ on the number of edges. This theorem was generalized for
doubling metrics by Bartal, Fandina, and Neiman \cite{BFN19}, who achieved the same bound 
via a much simpler construction.
Recently, Chang et al.\ \cite{CCL+24SoCG} improved the dependency on $\eps$ showing that every Euclidean point set admits a $(1+\eps)$-tree cover with $O(\eps^{1-d}\log(1/\eps))$ trees, matching the aforementioned bound on the number of edges of the $(1+\eps)$-spanners up to a polylogarithmic term in $1/\eps$.

\paragraph{Ultra-sparse regime.}
Although the aforementioned results provide a constant size bound, for $\eps=O(1)$ and $d = O(1)$, the constant hiding in the $O$-notation is quite large; even for the Euclidean plane, all known results incur a constant of at least 40.
The focus of this work is the ultra-sparse regime. 
It is only natural to seek constructions of spanners and tree covers where the size bound or number of trees approach those of a single spanning tree. 
In particular, the following basic question arises: what is the minimum number of trees or edges needed for achieving constant stretch?

Aronov et al.~\cite{ABCGHSV08} showed that for every $n$ and $k$ such that $0 \le k \le 2n-5$, every set of $n$ points in the Euclidean plane admits an $O(n/(k+1))$-spanner  with at most $n+k-1$ edges, where the constant in the $O$-notation is around 50. They also show that this tradeoff is tight up to a constant factor:  for every $n$ and $k$ such that $0 < k < n$, there is a set of $n$ points for which any spanner with $n+k-1$ edges must have stretch  at least $(2/\pi)\cdot \lfloor n/(k+1)\rfloor-1$.
Chew \cite{Chew89} proved that TD-Delaunay triangulation, which has at most $3n-6$ edges, achieves stretch at most $2$. 
Another family of spanners for Euclidean spaces are $\Theta_k$ graphs, which have at most $kn$ edges for a parameter $k \ge 1$. (See \Cref{def:theta_yao}.) This family was introduced independently by Keil and Gutwin \cite{Keil88,KG92}, and Clarkson \cite{Clarkson87}, who showed that for $k\ge 9$, the stretch of a $\Theta_k$ graph is at most  $1/(\cos\theta-\sin\theta)$ where $\theta = 2\pi/k$. The stretch bound has later been improved by \cite{RS91} to $1/(1 - 2 \sin(\theta/2))$ for $k\ge 7$. Current state of the art for general $k\ge 7$ is given by Bose et al. \cite{BCMRV16}. There is an ongoing research to pinpoint the exact stretch bounds for small values of $k$. For $k=6$, tight stretch is 2 \cite{BGHI10}. For $k=5$, the upper bound is 5.7 \cite{BHO21} and the lower bound is 3.78 \cite{BMRV15}. 
For $k=4$, the upper bound is  17 \cite{BCHS24} and the lower bound is 7 \cite{BBCRV13}.
For $k=3$, the stretch is unbounded \cite{Mol09}. A closely-related family of spanners are Yao$_k$ graphs, which also have at most $kn$ edges for $k \ge 1$. (See \Cref{def:theta_yao}.)
They were first shown to be spanners in \cite{ADDJS93}. 
Bose et al. \cite{BMNSZ04} showed that the stretch is at most $1/(\cos\theta - \sin\theta)$ for $k\ge 9$. The bound was subsequently improved by Barba et al. \cite{BBDFKORTVX15}, who also proved upper bounds of $5.8$ for $k=6$  and $2+\sqrt{3}$ for $k=5$. Yao$_4$ graphs have stretch at most 16.54 \cite{BHST24} and Yao$_3$ graphs have unbounded stretch \cite{Mol09}.

In contrast to spanners, the ultra-sparse regime for tree covers is much less understood.
Introduced in a different context, Chan's ``shifting lemma'' \cite{Cha98} yields a tree cover with 3 trees and stretch $6\sqrt{2}$. (Full details of this construction are given in \Cref{sec:shifting}.) For two trees, one can achieve stretch of $O(\sqrt{n\log{n}})$ \cite{BFN19}. On the negative side, it is known that a single tree incurs a stretch of $\Omega(n)$ \cite{Epp96,RR98}. Whether or not one can achieve a good (ideally constant) stretch with 2 trees remains an elusive open problem.

\begin{question}\label{q:few-trees}
What stretch can be achieved with two trees? In particular, can one get constant stretch?
\end{question}

\subsection{Our Contribution}
We answer \Cref{q:few-trees} in positive by presenting a surprisingly simple Steiner tree cover construction with 2 trees and stretch $\sqrt{26}$. Each of the trees in the construction is a quadtree where the internal vertices are the centers of the corresponding quadtree cells and the leaves are the input points. In particular, the internal vertices of the quadtree are Steiner points, which are not part of the input metric.
The result is summarized in the following theorem.

\begin{restatable}{theorem}{steiner}\label{thm:steiner}
Every set of points $P \subset \mathbb{R}^2$ admits a \emph{Steiner} tree cover consisting of two trees and with stretch $\sqrt{26}$.
\end{restatable}

Using the Steiner point removal result by Gupta~\cite{Gup01} in a black-box manner, we readily obtain a spanning tree cover with 2 trees and stretch $8\sqrt{26}$. Instead, we present a straightforward bottom-up Steiner point removal procedure, which allows us to obtain a spanning tree cover with stretch $4\sqrt{26}$. In particular, we prove the following theorem.
\begin{restatable}{theorem}{nonsteiner}\label{thm:non-steiner}
Every set of points $P \subset \mathbb{R}^2$ admits a \emph{spanning} tree cover consisting of two trees and with stretch $4\sqrt{26}$. 
\end{restatable}

Using an adaptation of the degree-reduction technique of \cite{CGMZ16} we can also bound the degree in both trees, while increasing the stretch bound by a factor of 2. For a parameter $\ell \ge 1$, we show a tradeoff of maximum degree of $1+3(\ell+1)$ and stretch $4\sqrt{26}/(1-2^{-\ell})$. In particular, for $\ell = 1$, we get a degree 7 and stretch $8\sqrt{26}$.

\begin{theorem}\label{thm:bdd}
Every set of points $P \subset \mathbb{R}^2$ admits a spanning tree cover with a constant stretch where each tree has a constant degree.
\end{theorem}

\Cref{thm:steiner} also implies a Steiner spanner with $4n-4$ edges and stretch $\sqrt{26}$. This improves the stretch bound of $\Theta_4$ and Yao$_4$ graphs, which are the sparsest in the family of Theta and Yao spanners and have at most $4n$ edges on all instances and at least $3n-7$ on some instances.
\Cref{thm:non-steiner} gives a spanner with $2n-2$ edges and stretch $4\sqrt{26}$. This stretch is strictly better than that of Aronov et al. \cite{ABCGHSV08} for the regime with the same number of edges. In comparison to the aforementioned Delaunay triangulations, which are planar graphs and have $3n-6$ edges for some instances, our construction has strictly smaller number of edges. Our construction is also sparser than Yao$_4$ and $\Theta_4$ graphs which can have at least $3n-7$ edges on some instances (see \Cref{sec:cones}).
In addition, all our constructions can be implemented in $O(n\log{n})$ time. The structural simplicity of our constructions makes them suitable for potential  applications, and we next give some examples.

\subsubsection{Applications}

The first application of our tree cover construction is a compact routing scheme, which is a distributed message passing algorithm formally introduced in \Cref{sec:routing}. 
There are many existing constructions of compact routing schemes for Euclidean and doubling metrics \cite{Tal04,Sli05,AGGM06,KRX06,KRXY07,GR08Soda,KRX08,ACGP10,KRXZ11,CGMZ16,CCL+24SoCG,CCLST25}. In particular, one can route in any $n$-point set in $\mathbb{R}^d$ with stretch $1+\eps$ and using $O_{\eps,d}(\log{n})$ bits of total memory; in some of these results, labels use as few as $\ceil{\log{n}}$ bits. Alas, even for a Euclidean plane and with large constant $\eps$, the the overall memory used in all of the previous constructions is at least $50\log{n}$ bits.
Our spanning tree construction from \Cref{thm:non-steiner} provides a constant stretch with a structurally simple overlay network---a union of two trees. The existing routing schemes on trees \cite{TZ01,GJL21} readily allow us to route in each of the trees using as few as $(1+o(1))\log{n}$ bits for labels and without routing tables. In other words, the total memory required for two trees is $(2+o(1))\log{n}$ bits. To be able to apply these results on top of our tree cover construction, the main challenge is to determine the right tree to route on without increasing the aforementioned memory bound by much. Our key technical contribution in terms of compact routing is a \emph{handshaking} protocol used to determine which tree to route on, which uses only additional $O(\log\log{n})$ bits. This is the first compact routing result for the Euclidean plane which achieves a constant stretch and uses as few as $(2+o(1))\log{n}$ bits of memory. In addition, we use a structurally simple overlay network consisting of two trees (and we route on top of only one), which could be of a potential practical interest.

Our next application is a path reporting distance oracle, which is a data structure in the centralized sequential model used to report paths between any two query points in the input point set. Specifically, we show that for any set $P$ of $n$ points in the Euclidean plane there is a data structure of size $O(n)$ which can be constructed in $O(n\log{n})$ time such that for every subsequent query, consisting of points $a,b\in P$, it outputs a path $P_{a,b}$ of stretch $4\sqrt{26}$ in time $O(|P_{a,b}|)$, where $|P_{a,b}|$ is the number of edges in $P_{a,b}$. See \Cref{sec:oracle} for details.   
It was previously shown that there is a data structure of size $O(n\log{n})$ which requires $O(n\log{n})$ time to be constructed; for every subsequent query it returns a $(1+\eps)$-approximate path in time linear in the path size (assuming $\eps$ is a constant) \cite{Eun20}. In comparison to this result, our data structure is of a smaller size. There are also constructions based on Yao and $\Theta$ graphs, which take $O(n\log{n})$ time to be computed, have size $O(n)$ and output a $(1+\eps)$-approximate path in time linear in the path size \cite{NS07}.  Compared to the Yao and $\Theta$ graph based data structures, our construction has a simpler structure --- it is a union of two trees augmented with additional information.

We also present a local routing algorithm, which is a navigation algorithm operating in the online setting. (See \Cref{sec:local-routing} for a formal definition.) We present a local routing algorithm for a set of $n$ points in $\mathbb{R}^2$ which uses an overlay network that consists of a pair of trees, and is thus of size $2n-2$, and achieves stretch of $4\sqrt{26}$. The key advantage of our local routing algorithm is the sparse (and simple) overlay network. In comparison to the best local routing algorithm for Yao$_4$ graphs, which achieve stretch 23.36 \cite{BHST24}, our construction achieves a better stretch and operates on a sparser overlay network. Although our stretch is inferior to the best known bound for $\Theta_4$ graph, which is 17 \cite{BCHS24}, our key advantage is the sparse overlay network.
We note, however, that the routing schemes based on Yao$_4$ and $\Theta_4$ require only knowledge of four neighbors at every stage of routing, since the out-degree of these constructions is bounded by 4. To address this point, we present another local routing algorithm which operates on the bounded degree tree cover construction, where each tree has degree of at most 7. The stretch of this algorithm is $8\sqrt{26}$ and it stores additional $2\log{n}$-bit information with each point.

\subsection{Proof overview}
In this section, we mainly focus on the
case where the input points constitute a \EMPH{$\sqrt{n}$ by $\sqrt{n}$ grid}. Not only is this a basic instance that allows us to convey the key ideas in a clean manner, but also our general construction is obtained through a reduction to the grid instance. For the grid instance, it is only natural to employ a recursive construction, where we partition the underlying grid square into 4 subsquares of half the side-length, recursively construct a 2-tree cover for each of the 4 subsquares, and then the remaining part is to carefully glue together the 2-tree covers that we get from the recursion for each of the 4 subsquares.
One can easily achieve constant stretch for a single recursion level (i.e., for a single distance scale), but it is unclear how to glue together the 2-tree covers from the lower level of the recursion to the current one without blowing up the stretch by a multiplicative constant factor.
In other words, a naive recursive construction will blow up the stretch by a multiplicative constant factor at each recursion level. With a more careful treatment, one can replace the multiplicative slack in the stretch at each level with an additive slack, but the number of levels is logarithmic in $n$, resulting in a stretch of $\approx \log n$. Moreover, for a general point set, the number of recursion levels will depend on the {\em aspect ratio} of the metric, which might lead to a stretch that is polynomial in $n$. Thus the crux is to come up with a 2-tree cover construction that \EMPH{breaks the dependency between different levels of the recursion}. 
This work presents a remarkably simple construction that achieves exactly that, as we describe next.

Our tree cover for the grid consists of a pair of  \EMPH{quadtrees}. Specifically, to construct the first tree, we take the square  $S$ underlying the grid and root it at (grid) point $o$ which is the center of $S$, as shown in \Cref{fig:construction}. Next, we partition $S$ into four equal subsquares ($A$, $B$, $C$, and $D$ in \Cref{fig:construction}) and invoke the recursion for each of them, connecting the root $o$ with the (grid points lying at the) centers of the corresponding squares. This process continues recursively and the recursion bottoms when a square with a single grid point is reached. See \Cref{fig:fullIntro} for an illustration of a single tree.

\begin{figure}[!htb]
\centering
\includegraphics[width=0.9\textwidth]{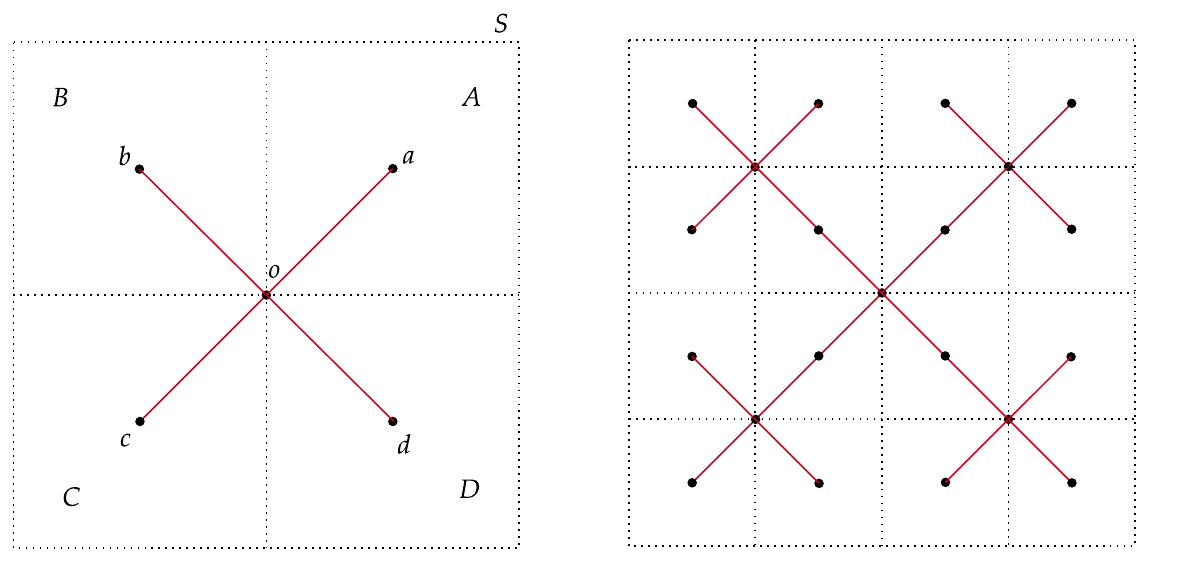}
\caption{Left: An illustration of one recursive call of the construction where $S$ is the square underlying the grid (or an axis-aligned square containing the input point set, for the Steiner construction). Square $S$ centered at $o$ is divided into four subsquares whose centers are connected to $o$. 
Right: An illustration of the tree after the recursive calls on each of the four subsquares of $S$.}
\label{fig:construction}
\end{figure}

\begin{figure}[!htb]
    \centering
    \begin{minipage}{0.5\textwidth}
        \centering
        \includegraphics[width=\linewidth]{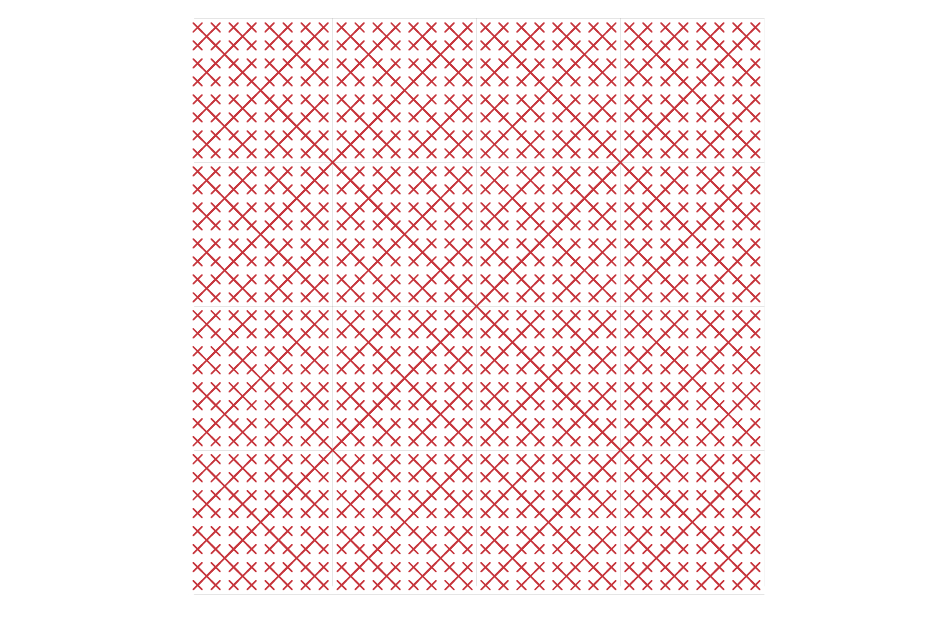}
        \caption{Full construction of a single tree.}
        \label{fig:fullIntro}
    \end{minipage}%
    \begin{minipage}{0.5\textwidth}
        \centering
        \includegraphics[width=0.65\linewidth]{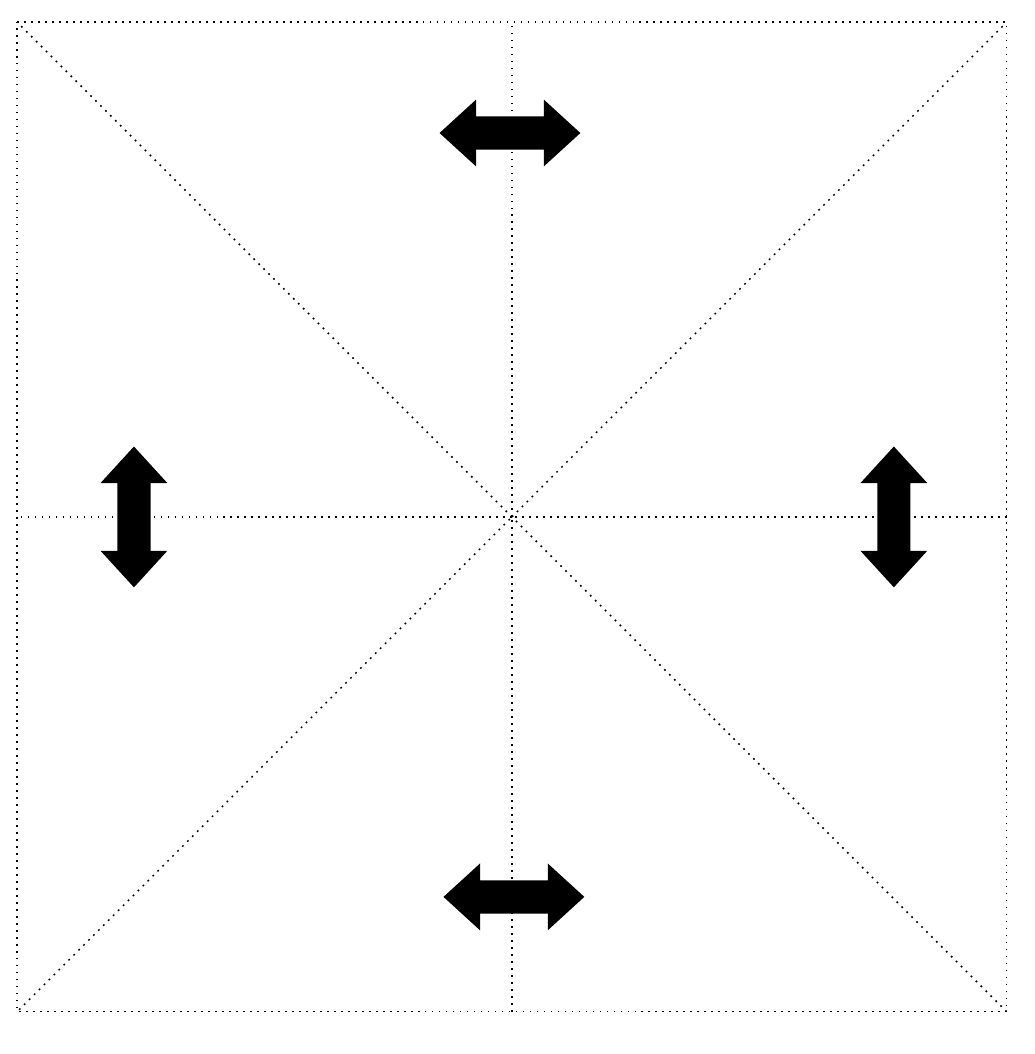}
        \caption{Regions with pairs having high stretch.}
        \label{fig:regionsIntro}
    \end{minipage}
\end{figure}

Consider the square $S$ containing all input points, which corresponds to the first recursion level of the construction, and the subdivision of $S$ into triangles, as shown in \Cref{fig:regionsIntro}.
Consider the pairs of points that are in the neighboring triangles labeled with an arrow. If a pair of points is close to the border between two triangles, their Euclidean distance is $O(1)$, while their tree distance is $\Theta(\sqrt{n})$. In other words, when restricting our attention to the pairs that lie in different subtrees of the tree rooted at the center of $S$, a constant fraction of pairs of points have stretch $\Theta(\sqrt{n})$. (For a general point set, the stretch could be as large as the aspect ratio.)
A large stretch is incurred also at lower levels of recursion, as the stretch bound decays by a factor of 2 at each level. We conclude that this tree has a large stretch for a constant fraction of pairs of points, which we call the ``bad pairs''.

Let us take the second tree to be a copy of the first tree, but {\em rotate} it by 45 degrees and {\em inflate} it by a factor of $\sqrt{2}$.
The resulting tree cover, with a {\em fractal-like} structure, is shown in \Cref{fig:full}. Specifically, the rotation of the second tree by $45$ degrees is done so that the pairs of neighboring triangles with a large stretch are disjoint between the two trees. The inflation by $\sqrt{2}$ simply ensures that all  input points are contained in the second tree. 
\begin{figure}[!htb]
\includegraphics[width=\textwidth]{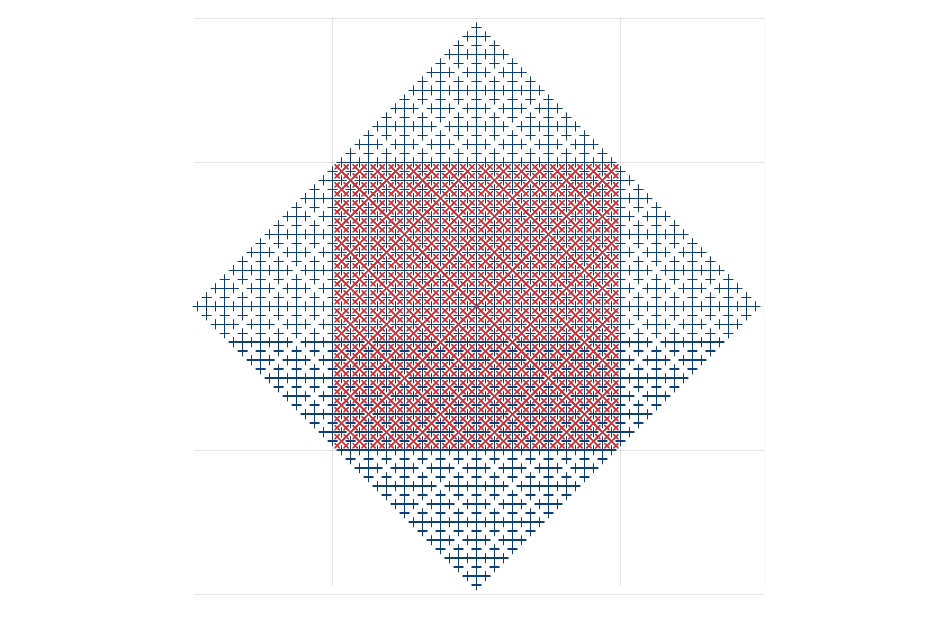}
\caption{The two trees in the tree cover. }
\label{fig:full}
\end{figure}

It is not difficult to verify that the second tree provides a good stretch for all bad pairs of the first tree in the first level of recursion (with stretch $\Theta(\sqrt{n})$).
\begin{figure}
    \centering
    \includegraphics[width=0.7\linewidth]{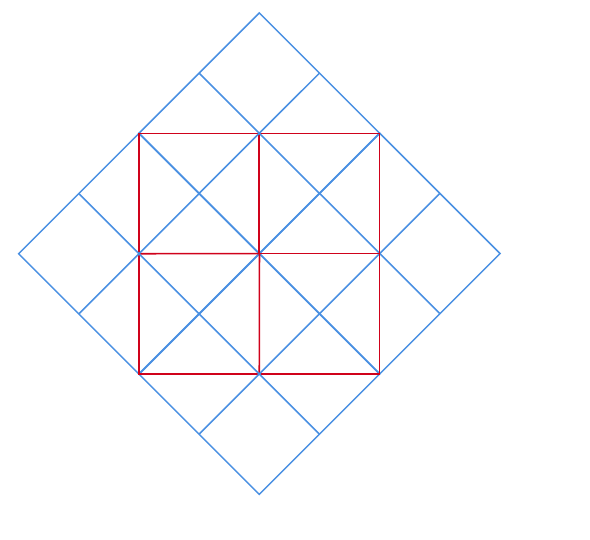}
    \caption{Fractal-like structure of our construction.}
    \label{fig:fractalIntro}
\end{figure}
It turns out that the specific fractal-like structure obtained by the rotation and inflation of the second tree enables us to prove a key structural property:
each quadtree cell of the first tree in {\em any recursion level} is surrounded by four quadtree cells in the second tree, exactly as is the case in the first recursion level. See \Cref{fig:fractalIntro} for an illustration. Using this structural property, we are able to completely get rid of the dependencies between different recursion levels in the stretch analysis and \EMPH{reason on each level separately}!  
Specifically, for every two points, we find the quadtree cell in the first tree such that these points lie in two different subsquares. If they lie in neighboring triangles and have a large stretch in the first tree, it can be shown that the 
construction for the corresponding quadtree cell in the second tree
gives a good stretch for them.

This construction for the grid instance extends seamlessly to general point sets, where \emph{Steiner points} assume the role of grid points. Specifically, 
our Steiner tree cover consists of a pair of {\em augmented} quadtrees, where
the square $S$ for the first tree is the smallest axis-aligned square containing the input point set,
the input points are the leaves in the quadtrees and the internal (Steiner) points assume the role of grid points, i.e., they are the centers of the quadtree cells.
We show that the stretch of our Steiner tree cover construction is $\sqrt{26}$ and give a \EMPH{tight} stretch analysis for this construction. 
A natural  question left open by this work is whether a different construction (even for the grid) could provide stretch better than $\sqrt{26}$---what is the optimal stretch achievable by 2 trees? Natural adaptations of our construction, such as using a square subdivision into more than eight triangles, do not seem to work. 

The full details of the Steiner construction can be found in \Cref{sec:steiner}. Our tight stretch analysis of the construction is given in \Cref{sec:tight}. To obtain a spanning tree cover, we employ a simple Steiner point removal procedure described in \Cref{sec:spr}, which blows up the stretch by a constant factor. Finally, to get a bounded degree spanning tree cover, we adapt a degree re-routing technique due to \cite{CGMZ16} (employed originally for the purpose of achieving a bounded degree $(1+\eps)$-spanner in doubling metrics) to our spanning tree cover construction; the details of this adaptation are given in \Cref{sec:bdd}. 

Selected applications of our tree cover are presented in \Cref{sec:applications};
for some of them, it is crucial to employ a {\em bounded degree} tree cover construction that does not contain any Steiner points.

\section{Tree cover construction}
In this section we describe the tree cover construction. The most basic construction contains Steiner points. In \Cref{sec:steiner} we describe the Steiner construction and provide a simple upper bound on stretch. In \Cref{sec:tight} we give a tight stretch analysis. The formal result is summarized in \Cref{thm:steiner}.
In \Cref{sec:spr} we prove \Cref{thm:non-steiner}. Specifically, we describe how to remove Steiner points so that the stretch grows by at most a factor of 4.

\subsection{Steiner tree cover construction}\label{sec:steiner}
Let $P$ be a set of points in the Euclidean plane and let $\delta$ be the minimal pairwise distance between points in $P$. Let $S_R$ be an an axis-aligned square containing $P$.  Let $S_B$ be the square obtained from $S_B$ by a $\pi/4$ rotation and scaling up by a factor of $\sqrt{2}$; note that the side centers of $S_B$ lie at the corners of $S_R$.
We invoke \textsc{Tree}$(S_R,P)$ (see \Cref{alg:construct}) below to obtain the \emph{red tree}, and  \textsc{Tree}$(S_B,P)$ to obtain the \emph{blue tree}. 
The procedure \textsc{Tree}$(S, P)$ for a given square $S$ and a set of points $P$ works as follows. Let $o$ be the center of $S$. The base case happens when $P$ contains a single point $p$. There, the output is the single edge between $s$ and $p$. Otherwise, $S$ is partitioned into four equal squares and $o$ is connected via an edge to each of the respective centers. Procedure recurses on each of the four squares. See \Cref{alg:construct} for details.
Square $S_R$ and all the squares in the subsequent recursive calls of \textsc{Tree}$(S_R,P)$ are called \emph{red squares}. Similarly $S_B$ and the squares in the subsequent recursive calls are called \emph{blue squares}.  Some of the statements in the subsequent proofs hold both for red and blue squares; we use \emph{recursion square} to denote either red or blue square in such statements.
See \Cref{fig:construction} for an illustration of the first few iterations of \textsc{Tree}$(S_R,P)$. 
Refer to \Cref{fig:full} for an illustration of the two trees after several levels of recursion.
\begin{algorithm}
\begin{algorithmic}[1]
\Procedure{Tree}{$S, P$}
\State let $o$ be the center of $S$
\If {$\diam(S) \le \delta/8$}
\If {$P \neq \emptyset$}
\State let $p$ be the point in $P$ \Comment{There can be only one point.}
\State \Return $\{o,p\}$
\Else
\State \Return $\emptyset$
\EndIf
\EndIf
\State partition $S$ into four equal squares: $A$, $B$, $C$, and $D$ centered at $a$, $b$, $c$, and $d$, respectively
\State $E \gets \{(o,a), (o,b), (o,c), (o,d)\}$
\State $E \gets E \cup \Call{Tree}{A, P \cap A} \cup \Call{Tree}{B, P \cap B}\cup \Call{Tree}{C, P \cap C}\cup \Call{Tree}{D, P \cap D} $
\State\Comment{If a point belongs to more than one square, break ties arbitrarily. }
\State\Comment{If $o \in P$, exclude it from the four squares.}
\State\Return $E$
\EndProcedure
\end{algorithmic}
\caption{$S$ is a square containing a point set $P$} 
\label{alg:construct}
\end{algorithm}

It is readily verified that the running time of this construction is $O(|P|\log\Phi)$, where $\Phi$ is the aspect ratio of $P$.
To avoid the dependence on the aspect ratio, we construct a compressed quadtree (with no single-child paths), where the number of vertices is at most $2|P|-1$. Such a quadtree is known to require $O(n \log n)$ time to be constructed \cite{SHP11}. Hence the running time of our construction is $O(n \log n)$.
\begin{lemma}\label{lem:time}
Given a point set $P$, a compressed version of the Steiner tree cover from \Cref{thm:steiner} can be computed in $O(|P| \log |P|)$ time. Each tree in the cover is without single-child paths and consists of at most $2|P|-1$ vertices.
\end{lemma}

The rest of this section is devoted to proving \Cref{thm:steiner}.
By \Cref{lem:correct}, each of the calls
\textsc{Tree}$(S_R,P)$ and  \textsc{Tree}$(S_B,P)$
yields a Steiner tree. \Cref{lem:stretch} asserts that the stretch of the tree cover obtained by these two trees is at most 8. (In \Cref{sec:tight}, we improve this bound to $\sqrt{26}$ and show that this is tight for our construction.)
Thus, it remains to prove 
\Cref{lem:correct} and
\Cref{lem:stretch}.

\begin{lemma}\label{lem:correct}
For any point set $P$ in $\mathbb{R}^2$ and for any square $S$ containing $P$,
Procedure \textsc{Tree}$(S,P)$ outputs a Steiner tree for $P$ whose edge set is contained in $S$.
\end{lemma}
\begin{proof}
The proof is by induction on the recursion level. In the base case, the procedure either returns a single edge inside $S$ or an empty set,
hence the statement clearly holds. For the induction step, we assume that the recursive calls produced four Steiner trees $T_A$, $T_B$, $T_C$ and $T_D$ whose edge sets are contained in $A$, $B$, $C$  and $D$, respectively. These trees are vertex-disjoint since the four squares only share the borders and the algorithm ensures that any point on the border between one or more squares is considered only in a single recursive call. The final tree is obtained by connecting $o$ to the roots of the four trees. This tree is fully contained in $S$ since $A$, $B$, $C$, and $D$ form a partition of $S$.
\end{proof}

The following invariant is used later on in the stretch analysis. Intuitively, it asserts the fractal shape of the construction. This is a simple geometric observation; see  \Cref{fig:fractal}
for an illustration.
\begin{figure}
    \centering
    \includegraphics[width=\textwidth]{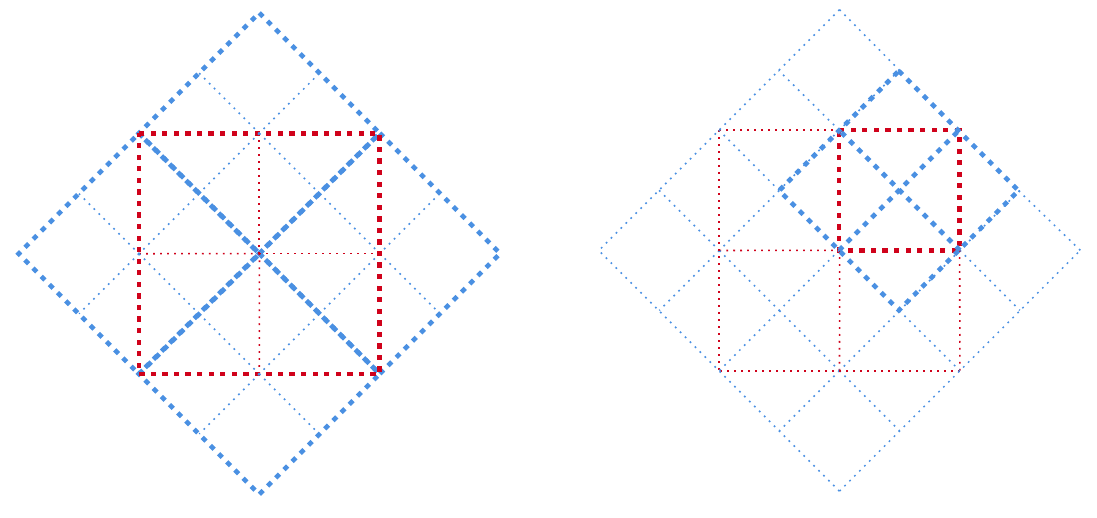}
    \caption{An illustration of \Cref{inv:fractal}. Left: For every red square there are four blue squares at its border. Right: The property holds for each of the four subsquares of the red square.}
    \label{fig:fractal}
\end{figure}
\begin{invariant}\label{inv:fractal}
For every red square $S$, there are four disjoint blue squares, each of which has a center at the middle point of a side of $S$ and two antipodal corners coinciding with the endpoints of the side.
\end{invariant}
\begin{proof}

The proof is by induction on the recursion level. For the base case, consider the red square $S_R$ containing the whole point set $P$. By the algorithm description, there is a blue square $S_B$ with side centers at the corners of $S_R$. This blue square is subsequently divided into four squares that satisfy the invariant. For the induction step, we assume that the statement holds for an arbitrary red square $S$. It is straightforward to verify that the statement holds for each of the four subsquares of $S$. See \Cref{fig:fractal} for an illustration.
\end{proof}

The following two claims are used in the subsequent stretch analysis.
\begin{claim}\label{clm:path}
Let $S$ be a recursion square centered at point $o$ such that $|P \cap S| \ge 2$ and let $p \in P \cap S$. Let $S'$ be a recursion square considered in one of the recursive calls in \textsc{Tree}$(S,P \cap S)$ such that $o \in S'$ and $p \in S'$. Then, the distance in $T$ between $o$ in $p$ is at most $d_T(o, p) \le \diam(S')$, where $T$ is the tree returned by \textsc{Tree}$(S,P \cap S)$.
\end{claim}
\begin{proof}
First observe that by construction there is an edge between $o$ and the center of $S'$, denoted by $o'$, with length $\norm{o - o'} = \diam(S')/2$.
We prove the claim by induction. For the base case we assume that $S'$ corresponds to a base case of the construction. Then, there is a direct edge between $o'$ and $p$ of length at most $\diam(S')/2$. We have $d_T(o,p) = \norm{o-o'} + \norm{o' - p} \le \diam(S')$. We proceed to prove the induction step. By construction, the first edge on the path is between $o$ and $o'$. The rest of the path is fully contained in one of the four subsquares of $S'$ and has length at most $\diam(S')/2$ by the induction hypothesis. Hence, the total length of the path is at most $\diam(S')$, as required.
\end{proof}

\begin{claim}\label{clm:dist-red}
Let $S$ be a recursion square centered at point $o=[x_o,y_o]$ and let $T$ be the tree obtained by invoking \textsc{Tree}$(S,P \cap S)$. Let $p \in P\cap S$ be an arbitrary point.
If $\norm{o-p} \le \delta/16$ then $d_T(o,p) \le \delta/8$; otherwise $d_T(o,p)\le 2\sqrt{2}\cdot \norm{o-p}$.
\end{claim}
\begin{proof}
Let $S'$ be the lowest-level recursion square considered in one of the recursive calls in \textsc{Tree}$(S,P \cap S)$.
If $\norm{o-p} \le \delta/16$, it is readily verified that $S'$ corresponds to a base case of the recursion,
hence $d_T(o,p)  \le \diam(S') \le \delta/8$, and we are done.
Otherwise, $\norm{o-p} > \delta/16$.
If $S'$ corresponds to a base case of the recursion,
then $d_T(o,p) \le \diam(S') \le \delta/8 < 2\norm{o-p}$.
Otherwise, $\norm{o-p} \ge \frac{\diam(S')}{2\sqrt{2}}$ since there is a smaller square with diameter $\diam(S')/2$, containing $o$ but not $p$.
By \Cref{clm:path}, $d_T(o,p) \le \diam(S') \le 2\sqrt{2}\cdot \norm{o-p}$.
\end{proof}

The following is the key technical lemma in our stretch analysis. In the subsequent proofs, we consider a recursion square $S$ that is partitioned into 4 subsquares and 8 triangles as shown in \Cref{fig:regions}. 
\begin{figure}[ht]
\begin{centering}
\includegraphics[width=\textwidth]{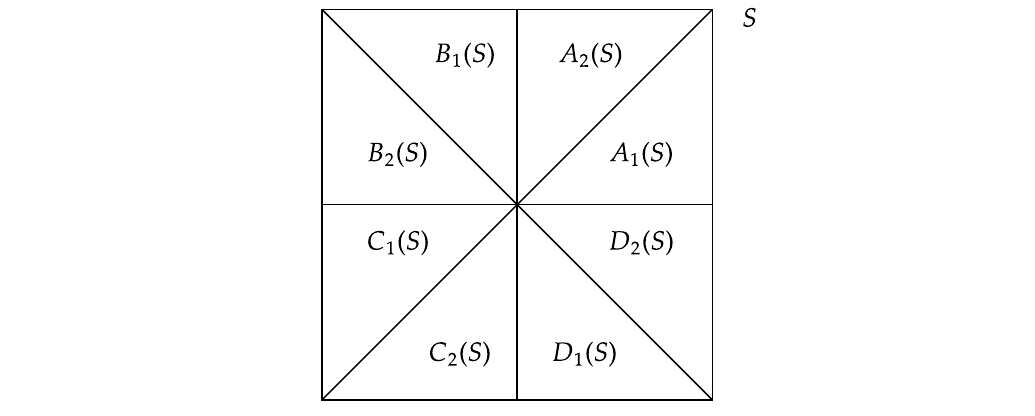}
\end{centering}
\caption{A recursion square $S$ partitioned into 8 triangles.}
\label{fig:regions}
\end{figure}

\begin{lemma}\label{lem:cones}
Let $S$ be a recursion square and let $T$ be the corresponding tree. The stretch in $T$ between any two points belonging to non-consecutive triangles of $S$ is at most $8$.
\end{lemma}
\begin{proof}
Let $a$ and $b$ be an arbitrary pair of points in any two non-consecutive triangles of $S$ and write $d = \norm{a-b}$. Let $o$ denote the center of $S$.

We first consider the case where $\norm{o-a} \le \delta/16$ or $\norm{o-b} \le \delta/16$. Note that it cannot hold that both $\norm{o-a}$ and $\norm{o-b}$ are at most $\delta/16$, since by triangle inequality that would yield $\norm{a-b} \le \delta/8$, which is a contradiction to the definition of $\delta$. Without loss of generality assume that $\norm{o-a} \le \delta/16$;
the case where $\norm{o-b} \le \delta/16$ is symmetric. 
By construction, the distance between $a$ and $b$ in $T$ can be upper bounded as follows. 
\begin{equation} \label{dtab}
    d_T(a,b) ~=~ d_T(o,a) + d_T(o,b) ~\le~ \frac{\delta}{8} + 2\sqrt{2}\cdot\norm{o-b},
\end{equation}
where the last inequality follows from \Cref{clm:dist-red}. We know that $\norm{o-b} \ge \norm{a-b} - \norm{o-a} \ge \delta - \delta/16 = \frac{15\delta}{16}$ by the triangle inequality and the definition of $\delta$. Applying the triangle inequality again, we have the following lower bound on $d$.
\begin{equation} \label{lowerd}
    d ~=~ \norm{a-b} ~\ge~ \norm{o-b} - \norm{o-a} ~\ge~ \norm{o-b} - \delta/16
    \end{equation}
\Cref{dtab} and \Cref{lowerd} readily give us the upper bound on the stretch. 
\begin{align}\label{eq:derivative}
\frac{d_T(a,b)}{d} \le  \frac{\frac{\delta}{8} + 2\sqrt{2}\cdot \norm{o-b}}{\norm{o-b}- \frac{\delta}{16}} 
\end{align}
Consider the function $f(x) = \frac{\frac{\delta}{8} + 2\sqrt{2}x}{x - \frac{\delta}{16}}$ and observe that it is monotonically decreasing for all $x > \frac{\delta}{16}$, since $f'(x) = -\frac{32(1 + \sqrt{2})\delta}{(16x - \delta)^2} < 0$. Since $\norm{o-b} \ge \frac{15\delta}{16}$, we have $f(\norm{o-b}) \le f\left(\frac{15\delta}{16}\right)$. This allows us to upper bound the right-hand side in \Cref{eq:derivative}.
\begin{align*}
f(\norm{o-b}) \le f\left(\frac{15\delta}{16}\right) \le \frac{\frac{\delta}{8}+2\sqrt{2}\cdot \frac{15\delta}{16}}{\frac{15\delta}{16}-\frac{\delta}{16}} < 4
\end{align*}

In conclusion, when $\norm{o-a} \le \delta/16$, the stretch is less than 4.

We proceed to handle the case where $\norm{o-a} > \delta/16$ and $\norm{o-b} > \delta/16$. Let $\alpha \coloneqq \norm{o-a}$ and $\beta \coloneqq \norm{o-b}$. 
By \Cref{clm:dist-red},
$d_T(a,b) = d_T(o,a) + d_T(o,b) \le 2\sqrt{2}\cdot(\alpha+\beta)$.
Denote by $\phi$ the angle between $oa$ and $ob$ and note that $\phi \ge \pi/4$ since $a$ and $b$ are in non-consecutive triangles of $S$.
By the law of cosines, it follows that $d = \sqrt{\alpha^2+\beta^2-2\alpha\beta \cos \phi} \ge \sqrt{\alpha^2+\beta^2-\alpha\beta\sqrt{2}}$. 
Writing $\beta = \gamma\cdot \alpha$, we obtain the following upper bound on the stretch.
\begin{align*}
\frac{d_T(a,b)}{d} &\le \frac{2\sqrt{2}\cdot(\alpha+\beta)}{\sqrt{\alpha^2+\beta^2-\alpha\beta\sqrt{2}}} = \frac{2\sqrt{2}\cdot (1+\gamma)}{\sqrt{1+\gamma^2-\gamma\sqrt{2}}} 
\le 2\sqrt{2}\cdot\sqrt{1 + \frac{2+\sqrt{2}}{\gamma + \frac{1}{\gamma}- \sqrt{2}}}\\
&\leq 2\sqrt{2}\cdot\sqrt{1 + \frac{2+\sqrt{2}}{2-\sqrt{2}}} = 4\sqrt{2 + \sqrt{2}} < 8
\end{align*}
\end{proof}

\begin{corollary}\label{cor:diagonal}
Let $S$ be a recursion square and let $T$ be the corresponding tree. The stretch in $T$ between any two points  
belonging to consecutive triangles and the same subsquare  of $S$ is at most $8$.
\end{corollary}
\begin{proof}
There are only four pairs of consecutive triangles that belong to the same subsquare of $S$, namely $\{A_1(S), A_2(S)\}$, $\{B_1(S), B_2(S)\}$, $\{C_1(S), C_2(S)\}$, and  $\{D_1(S), D_2(S)\}$. By symmetry, it suffices to focus on one pair, say $\{A_1(S), A_2(S)\}$. Let $a$ and $b$ be a pair of arbitrary points in $A_1(S)$ and $A_2(S)$, respectively. Let \textsc{Tree}$(S', P\cap S')$ be the call recursively invoked by \textsc{Tree}$(S, P\cap S)$ 
such that $a$ and $b$ are in different subsquares of $S'$. Then, $a \in A_1(S') \cup D(S') \cup C_2(S')$ and $b \in A_2(S') \cup B(S') \cup C_1(S')$. Since $a$ and $b$ are in different subsquares of $S'$ it is impossible that $a \in A_1(S')$ and $b \in A_2(S')$ or $a \in C_2(S')$ and $b \in C_1(S')$. See \Cref{fig:diagonal} for an illustration.
\begin{figure}
\centering
\includegraphics[width=0.78\linewidth]{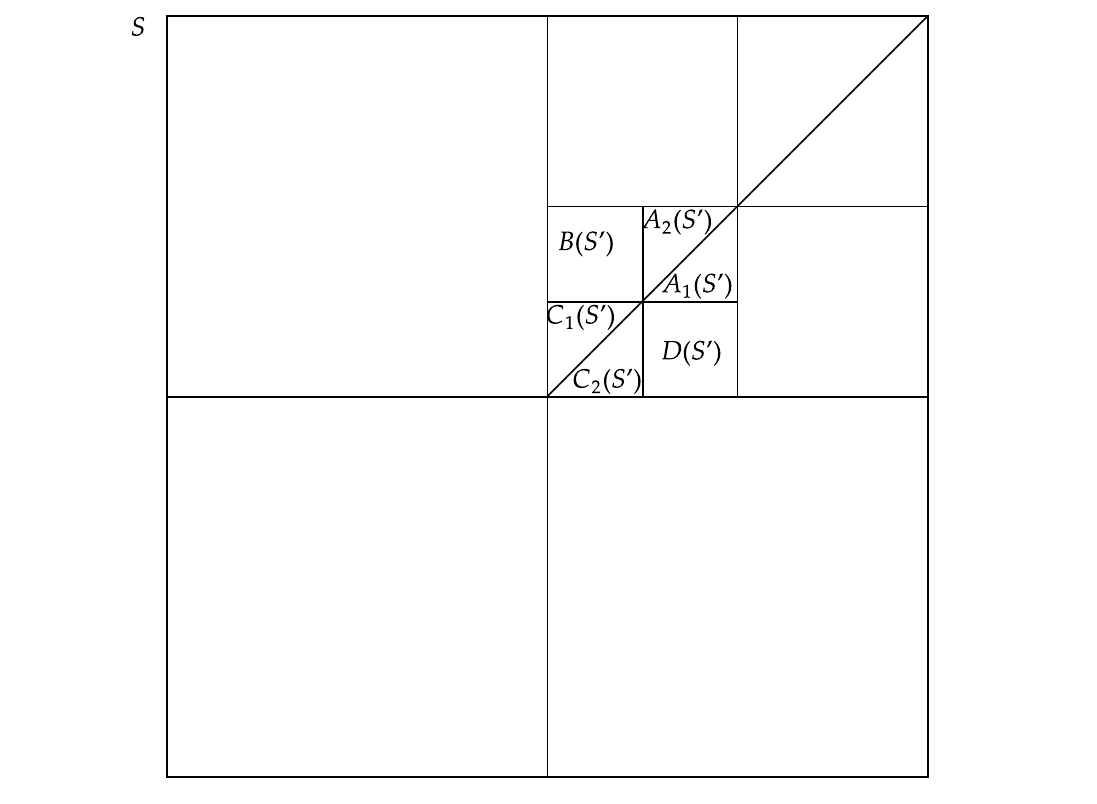}
\caption{An illustration of the proof of \Cref{cor:diagonal}. }
\label{fig:diagonal}
\end{figure}
It follows that $a$ and $b$ belong to non-consecutive triangles of $S'$, hence
the stretch between $a$ and $b$ with respect to $T$ is  at most 8 by \Cref{lem:cones}.
\end{proof}

\begin{lemma}\label{lem:stretch}
The stretch of the tree cover is at most 8. 
\end{lemma}
\begin{proof}
Consider an arbitrary pair of points $a$ and $b$ and let $S_1$ be the red square such that $a$ and $b$ belong to different sub-squares of $S_1$. See \Cref{fig:final} for an illustration. In the case where $a$ and $b$ are in non-consecutive triangles of $S_1$, the stretch between $a$ and $b$ is at most $8$ by \Cref{lem:cones}. 

We henceforth assume that $a$ and $b$ belong to consecutive triangles of $S_1$; without loss of generality, $a \in A_2(S_1), b \in B_1(S_1)$. By \Cref{inv:fractal}, there is a blue square region $S_2$ centered at the middle of the top side of $S_1$. (See \Cref{fig:final} for an illustration.) Because of their placement in the red square, $a$ is either in $C_2(S_2)$ or $D_1(S_2)$ and $b$ is either in 
 $D_2(S_2)$ or $A_1(S_2)$. If $a$ and $b$ are in non-consecutive triangles of $S_2$, then by \Cref{lem:cones}, the stretch is at most $8$. Otherwise, $a\in  D_1(S_2)$ and $b\in D_2(S_2)$, and in particular, they are in consecutive triangles of the same subsquare of $S_2$, hence the stretch between $a$ and $b$ in the blue tree is at most $8$ by \Cref{cor:diagonal}.
\end{proof}

\begin{figure}[ht]
\includegraphics[width=\textwidth]{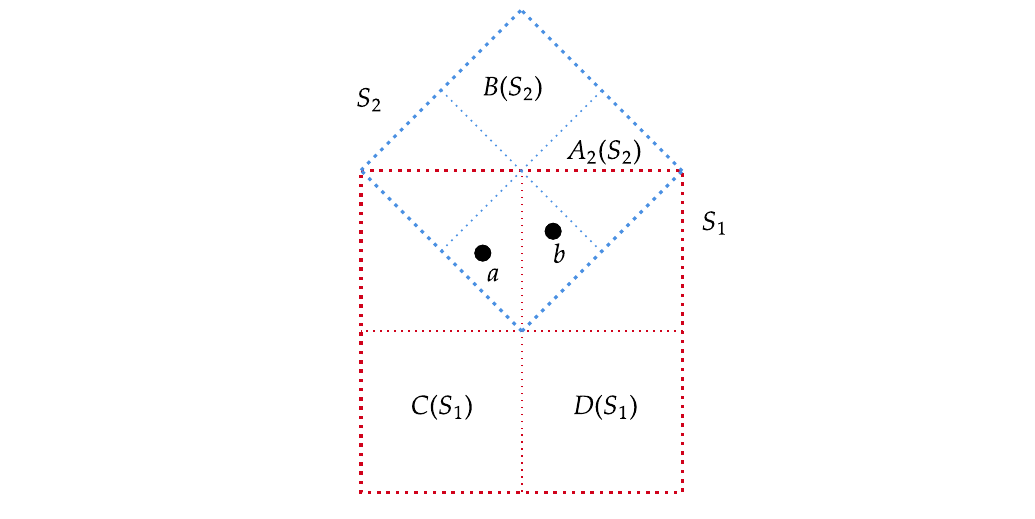}
\caption{An illustration of \Cref{lem:stretch}.}
\label{fig:final}
\end{figure}

The proof of \Cref{lem:stretch} implies the following observation, which is be useful for the applications presented in \Cref{sec:applications}.
\begin{observation}\label{obs:choose-tree}
Consider an arbitrary pair of points $a$ and $b$ and let $S_1$ be the red square such that $a$ and $b$ belong to different sub-squares of $S_1$. If $a$ and $b$ belong to non-consecutive triangles of $S_1$, the stretch between $a$ and $b$ is preserved in the red tree. Otherwise, the stretch between $a$ and $b$ is preserved in the blue tree.
\end{observation}
\subsection{A tight stretch analysis of the Steiner construction}\label{sec:tight}
This section is dedicated to improving the upper bound on the stretch of the tree cover construction. It suffices to improve the bound from \Cref{lem:cones}.
\begin{lemma}\label{lem:tight}
Let $S$ be a recursion square and let $T$ be the corresponding tree. The stretch in $T$ between any two points belonging to non-consecutive triangles of $S$ is at most $\sqrt{26}$.
\end{lemma} 
\begin{proof}
From the proof of \Cref{lem:cones}, we infer that it suffices to focus on the case when $\norm{o-a} > \delta/16$ and $\norm{o-b} > \delta/16$. Consider first the case where there are at least two triangles separating the triangles of $a$ and $b$. Then, $d_T(a,b) \le 2\sqrt{2}\cdot(\alpha+\beta)$ and $d = \sqrt{\alpha^2+\beta^2 - 2\alpha\beta\cos{\phi}} \ge \sqrt{\alpha^2+\beta^2}$, because $\pi \ge \phi \ge \pi/2$. The stretch is upper bounded as follows.
\begin{align*}
\frac{d_t(a,b)}{d} \le \frac{2\sqrt{2}\cdot(\alpha+\beta)}{\sqrt{\alpha^2+\beta^2}} \le \frac{2\sqrt{2}\cdot(\alpha+\alpha\gamma)}{\sqrt{\alpha^2(1+\gamma^2)}} = 2\sqrt{2}\cdot\sqrt{1+\frac{2}{\gamma + \frac{1}{\gamma}}} =4
\end{align*}

The rest of the proof is dedicated to the case when there is only one cone between the cone of $a$ and $b$. Without loss of generality suppose that $S$ is centered at $o=(0,0)$ and has side length 1. We consider the case where $a \in A_1(S)$ and $b \in B_1(S)$. The other cases are implied by symmetry. \Cref{clm:path} implies that $d_T(o,b) \le \frac{\sqrt{2}}{2}$.
We define $\level(a)$ as follows. Let $S'$ be the smallest square corresponding to a recursive call of \textsc{Tree} which contains $a$ and has $o$ as a corner. If $S' = A(S)$, then we define $\level(a) \coloneqq 1$. If $S'$ is one of the four subsquares of $A(S)$, then we define $\level(a) \coloneqq 2$, and so on. We define $\level(b)$ analogously. It is sufficient to consider two cases: first when $\level(a) = 1$ and $\level(b) \ge 1$ and second when $\level(b) = 1$ and $\level(a) \ge 1$. The other cases are implied by scaling --- if there is an integer $i$ such that $\level(a) \ge i$ and $\level{b} \ge i$, then we can define $S$ to be the square centered at $o$ and with side length $2^{-i+1}$.
For the first case, consider the regions $X$, $Y$, $Z$, and $W$ within $A_1(S)$ as shown in \Cref{fig:cases}. If $a \in X \cup Y$, then by \Cref{clm:path}, $d_T(o', a) \le \frac{\sqrt{2}}{8}$. We have $d_T(a,b) \le \frac{7\sqrt{2}}{8}$ and $\norm{a-b} \ge \frac{1}{4}$ and the stretch is less than 5. If $a \in Z$, then $d_T(a,b) \le \sqrt{2}$ and $\norm{a-b} \ge \sqrt{\left(\frac{1}{4}\right)^2 + \left(\frac{1}{8}\right)^2} = \frac{\sqrt{5}}{8}$; the stretch is less than $\frac{8\sqrt{2}}{\sqrt{5}} < \sqrt{26}$. Otherwise, since $\level(a) = 1$, $a$ must be in $W$. We have $\norm{a-b} \ge \frac{3}{8}$ and $d_T(a,b) \le \sqrt{2}$; the stretch is less than 4.
\begin{figure}
    \centering
    \includegraphics[width=0.78\linewidth]{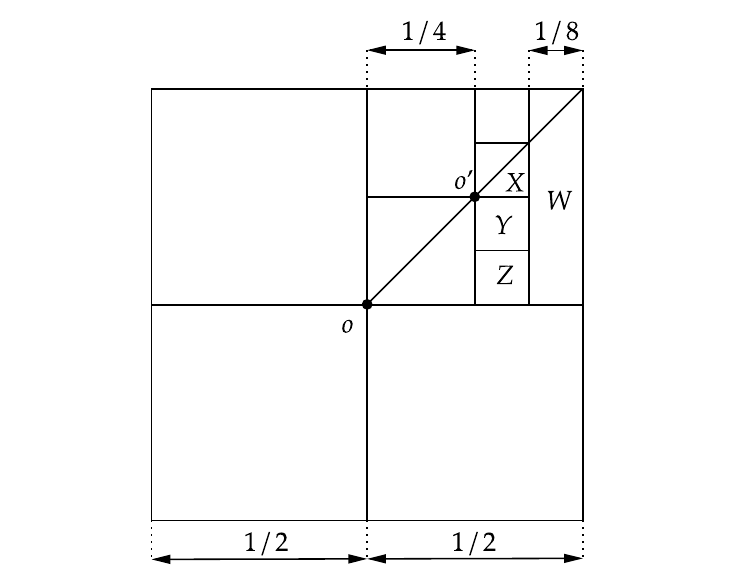}
    \caption{Regions used in the case analysis in \Cref{lem:tight}.}
    \label{fig:cases}
\end{figure}

The rest of the proof is dedicated to the second case, where $\level(b) = 1$ and $\level(a) \ge 1$. We have $d_T(o,b) \le \frac{\sqrt{2}}{2}$. The point $b\in B_1(S)$ closest to $A_1(S)$ with $\level(b) = 1$ has coordinates $b=\left(0, \frac{1}{4} \right)$. To finalize the analysis, we need the following claim.
\begin{claim}\label{clm:slope}
Let $a$ be a point in $A_1(S)$. If $z\sqrt{2} \le  d_T(o,a) $ for $z \in \left[0,\frac{1}{2}\right]$, then $a$ lies below the line $y=\frac{3}{2}x-\frac{z}{2}$.
\end{claim}
\begin{proof}
Consider the following iterative halving.
Initialize $\ell = 0$ and $r=\frac{1}{2}$ and consider the square $S'$ with antipodal corners at $(\ell, \ell)$ and $(r,r)$. Define $m = \frac{\ell+r}{2}$. Note that $(m,m)$ is the center of $S'$. Consider two cases.
\begin{figure}
    \centering
    \includegraphics[width=0.9\linewidth]{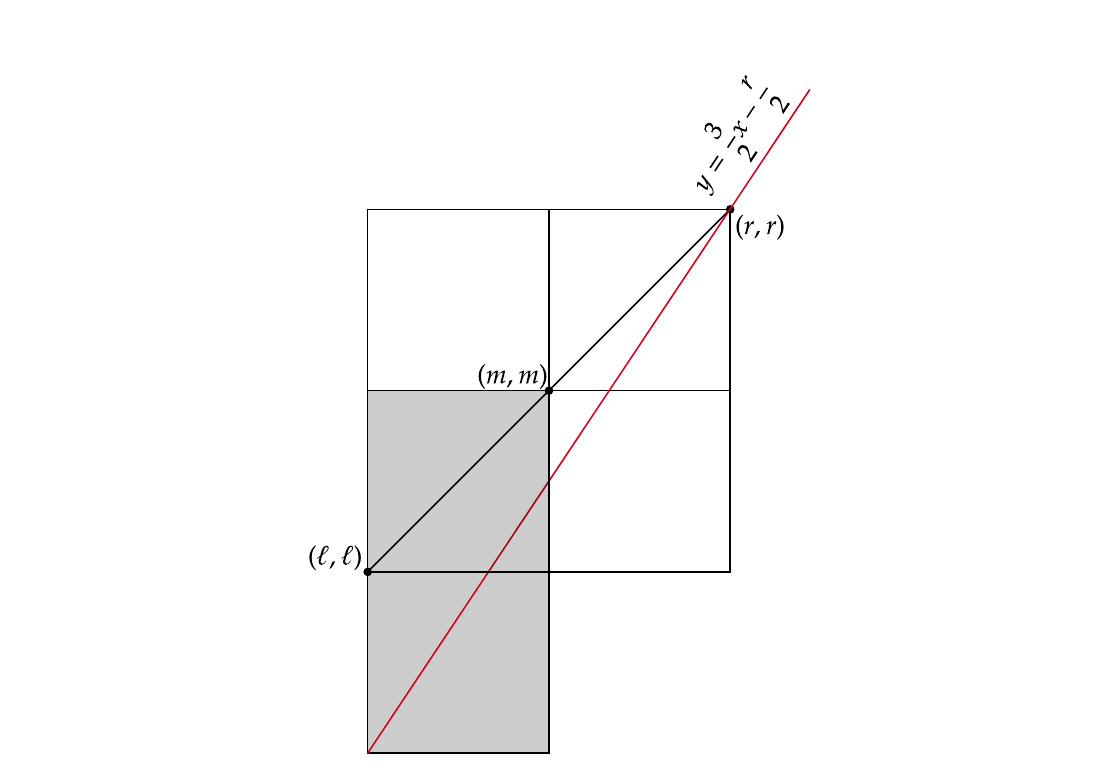}
    \caption{An illustration of the proof of \Cref{clm:slope}.}
    \label{fig:slope}
\end{figure}

\noindent\textit{Case 1. $m \le z$.}
By \Cref{clm:path}, every point $a = (x_a, y_a)$ with $x_a \in [\ell, m]$ and $d_T(o,a) \ge m$ have $y_a \ge \ell - (m - \ell) = 2\ell - m = \frac{3}{2}\ell - \frac{r}{2}$. In particular all such points $a$ lie below the line $\frac{3}{2}x - \frac{r}{2}$. Set $S'$ to be the square with antipodal corners at $(m, m)$ and $(r,r)$ and set $\ell = m$. Proceed with the iterative halving.

\noindent\textit{Case 2. $z< m$.} Set $S'$ to be the square with antipodal corners at $(\ell, \ell)$ and $(m, m)$ and set $r = m$. Proceed with the iterative halving.

The iterative halving ends when $S'$ corresponds to the lowest level of recursive construction. Throughout the halving we maintained $(z,z) \in S'$. Every point $a = (x_a, y_a)$ with $x_a \in [0, \ell]$ and $d_T(o,a) \ge \ell$ lie below $\frac{3}{2}x - \frac{r}{2}$. By construction, there is only a single point in $S'$. This point lies below $\frac{3}{2}x - \frac{z}{2}$ and the claim follows.
\end{proof}
Let $z$ be as in \Cref{clm:slope}. The distance between the line $y=\frac{3}{2}x-\frac{z}{2}$ and $b=\left(0, \frac{1}{4}\right)$ is as follows. 

\begin{align*}
    d = \frac{\abs{-\frac{1}{4}-\frac{z}{2}}}{\sqrt{\left(\frac{3}{2}\right)^2+1}}=\frac{1+2z}{2\sqrt{13}}
\end{align*}
Since $a$ lies below the line $\frac{3}{2}x-\frac{z}{2}$, it must be that $\norm{a-b}\ge d$. The stretch is upper bounded as follows.
\begin{align*}
\frac{d_T(a,b)}{\norm{a-b}} \le \frac{d_T(o,a)+d_T(o,b)}{d} \le \frac{z\sqrt{2}+\frac{\sqrt{2}}{2}}{\frac{1+2z}{2\sqrt{13}}} = \sqrt{26} 
\end{align*}
\end{proof}

\begin{observation}
The stretch of the construction is at least $\sqrt{26}$.
\end{observation}
\begin{proof}
Consider a square $S$ centered at $o=(0,0)$ and let $b = \left(0,\frac{1}{4}\right)$, such that $b \in B_1(S)$. Then, $d_T(o,b) = \frac{\sqrt{2}}{2}$.  Let $a = \left(\frac{3}{20},\frac{3}{20}\right)$. Then, $d_T(o,a) = \frac{3\sqrt{2}}{20}$ and $\norm{a-b} = \frac{\sqrt{13}}{20}$. The stretch is as follows.
\begin{align*}
    \frac{d_T(a,b)}{\norm{a-b}} = \frac{\frac{13\sqrt{2}}{20}}{\frac{\sqrt{13}}{20}} = \sqrt{26}
\end{align*}
\end{proof}
\subsection{A simple Steiner point removal procedure}\label{sec:spr}

In this section we prove \Cref{thm:non-steiner}. 

We first note that one can prune the Steiner points from our Steiner tree cover construction by simply applying the Steiner point removal (SPR) procedure of Gupta \cite{Gup01} as a black-box, which would blow up the stretch by a factor of 8. 
Instead, we present a direct, simple linear-time SPR procedure, which blows up the stretch by only a factor of 4. It is possible to further improve the factor via a more careful pruning procedure; however, our main focus in this work is the simplicity of the constructions rather than optimizing the exact constant in the stretch bound.

Given a $T$ constructed by \textsc{Tree}$(S, P)$, we describe a simple bottom-up recursive procedure that constructs a rooted spanning tree $(T',r')$ for $P$, such that the pairwise distance in $T'$ between any pair of points in $P$ is at most 4 times larger than their distance in the original Steiner tree $T$. 

The recursion bottoms at recursion squares of diameter at most $\delta/8$, each containing at most one input point. If such a square $S$ contains point $p$, then the returned tree consists of $p$ as its only vertex, designated as the root. Otherwise, $S$ is empty and the returned tree is the empty tree. 

For the recursion step, consider any recursion square $S$ of diameter $\Delta_i > \delta/8$, which is partitioned into the four equal squares $A,B,C,D$, and let $(T'_A,r'_A),(T'_B,r'_B),(T'_C,r'_C),(T'_D,r'_D)$ be the respective rooted spanning trees for $P \cap A, P\cap B, P \cap C, P\cap D$ returned from the recursion. Note that some of these four trees may be empty; w.l.o.g.\ assume that $T'_A$ is nonempty.
The spanning tree $(T',r')$ of $P \cap S$ is constructed by assigning $r' = r'_A$ as the root of $T'$ and connecting $r'$ with direct edges to the roots $r'_B,r'_C,r'_D$ 
of $T'_B,T'_C,T'_D$ (if one of these trees is empty, no edge connecting $r'$ to it is added).

It can now be observed inductively that \begin{equation} \label{eq:base} \forall p \in P \cap S: d_{T'}(r',p) \le 4 \cdot d_T(o,p),\end{equation} 
where $T$ is the original Steiner tree for $P \cap S$ 
and $o$ is the center of $S$, and thus also the root of $T$. The induction basis is trivial. The induction step follows from the induction hypothesis for each of the four subtrees of $T'$ versus $T$, in conjunction with the observation that the weights of the edges between $r'$ to the roots of $T'_B,T'_C,T'_D$ exceed those of the edges between the root $o$ of $T$ and the roots of the respective subtrees of $T$, which are the centers of the subsquares $A,B,C,D$, by at most a factor of 4.

Now consider any pair  $p,q$ of input points and their least common ancestor $l$ in the original Steiner tree $T$.
Let $(T_l,l)$ be the subtree of $T$ rooted at $l$ and let $(T'_l,r'_l)$ be the corresponding spanning tree.
By construction we have $d_T(p,q) = d_{T_l}(l,p) + d_{T_l}(l,q)$ and $d_{T'}(p,q) \le d_{T'_l}(r'_l,p) + d_{T'_l}(r'_l,q)$.
\Cref{eq:base} now yields 
$$d_{T'}(p,q) ~\le~ d_{T'_l}(r'_l,p) + d_{T'_l}(r'_l,q) ~\le~ 4\cdot d_{T_l}(l,p) + 4 \cdot d_{T_l}(l,q) ~=~ 4\cdot d_T(p,q),$$
which completes the proof of \Cref{thm:non-steiner}.

The following observation will be used in the bounded degree construction.

\begin{observation}\label{obs:spr}
Consider an arbitrary step of the Steiner point removal procedure and let $r'$ be the chosen root. If $r'_B$ (or likewise $r'_C$, or $r'_D$) are connected to some child of $r'$ (instead of $r'$ itself) such that $d_T'(r'_B, r') \le \gamma\Delta_i$, then the Steiner point removal procedure has a stretch blowup of $4\gamma$ (instead of 4).
\end{observation}

For the application to local routing (see \Cref{sec:local-routing}), we shall use the following property.
\begin{observation}\label{obs:navigation}
Let $T'$ be a non-Steiner tree from the construction in \Cref{thm:non-steiner}. Consider any recursion square $S$ corresponding to $T'$ and let $r$ be the root of the subtree of $T'$ corresponding to $S$. Then:
\begin{enumerate}
    \item If $r$ has a parent, it is outside of $S$. All the other vertices in the subtree of $T'$ rooted at $r$ have their parents inside $S$.
    \item Vertex $r$ has at most one child in each of the three subsquares of $S$ where it does not belong to.
\end{enumerate}
\end{observation}

\subsection{Bounded degree construction}\label{sec:bdd}
In this section we show how to modify the spanning tree cover construction so that each tree has a degree bounded by an absolute constant. The degree reduction technique is adapted from \cite{CGMZ16}.

We describe the procedure on a tree $T'$ obtained by the Steiner point removal procedure from a tree $T$.
Let $\ell$ be an arbitrary constant, to be set later. Orient every edge in $T'$ towards the root. For every level $i$ and every node $u$, let $M_i(u)$ be the set of children of $u$ added at the level $i$. Let $I_u = \{i_1, i_2,\ldots \}$ be the set of levels where $M_i(u)$ is nonempty. We use $M_j^u \coloneqq M_{i_j}(u)$ to simplify the notation. 

We modify the arcs going into $u$ as follows. For $1 \le j \le \ell$, keep the arcs from $M_j^u$ to $u$. For $j > \ell$, pick a vertex $w \in M_{j-\ell}^u$ and for each point $v \in M_j^u$ replace $(v,u)$ by $(v,w)$. This concludes the edge modification process. Let $\tilde{T'} = (P,\tilde{E'})$ be the resulting tree. Since every edge is replaced by exactly another edge, the overall number of edges does not change.

\begin{lemma}
Every point in $P$ has a degree at most $1+3(\ell+1)$ in $\tilde{T'}$.
\end{lemma}
\begin{proof}
Consider a point $u \in P$. The edges incident on $u$ fall into the following three categories. First, there is at most 1 edge directed out of $u$. Out of the edges directed into $u$, at most $3\ell$ remain. There are potentially some new edges added into $u$. Consider an edge $(u,v)$ directed out of $u$, belonging to some $M_j^v$. There can be at most $|M_{j+\ell}^v|\le 3$ edges directed into $u$ due to this edge. Overall, the total degree of $u$ after the modification procedure is $1+3(\ell+1)$.
\end{proof}

The following observation is the key ingredient for our compact routing scheme described in \Cref{sec:routing}.
\begin{observation}\label{obs:hierarchy}
Every vertex in $P$ has children in $\tilde{T'}$ in at most $\ell+1$ different levels.
\end{observation}
\begin{proof}
Consider an arbitrary vertex $u \in P$.
After the degree reduction step, $u$ can be a parent of vertices at levels $1\le j \le \ell$ and at most once for the vertices in $M_{j+\ell}^v$ for some ancestor $v$ of $u$.
\end{proof}
We next show that this modification does not increase the stretch by much.
\begin{lemma}
For every pair of vertices $(u,v)$ in $\tilde{T'}$ their distance  with respect to the new edge set $\tilde{E}$ satisfies $\tilde{d}(v,u) \le 4d_T(u,v)/(1-2^{-\ell})$, where $d_T(v,u)$ is their distance in the original Steiner tree.
\end{lemma}
\begin{proof}
We focus on an edge $(v,u)$ such that $v$ is a child of $u$ and $(v,u)$ is replaced in the construction. Let $\Delta_j$ be the diameter of the square which $u$ is the root of. We shall show that $\tilde{d}(v,u) \le \frac{\Delta_j}{1-2^{-\ell}}$, which by \Cref{obs:spr} implies that the stretch is increased by a factor of $4/(1-2^{-\ell})$ with respect to the original Steiner tree.

By construction $v \in M_j^u$ for some $j > \ell$. Let $v_0 = v$. We also have that for  $0 \le s \le s_j\coloneqq \ceil{\frac{j-1}{\ell}}$, there is a $v_s \in M_{j-s\ell}^u$ such that $(v_s, v_{s+1}) \in \tilde{E'}$ and $(v_{s_j},u) \in \tilde{E'}$. Then, there is a path from $v$ to $u$ in $\tilde{E'}$ of the form $v=v_0, v_1, \ldots, v_{s_j},u$. The length of this path is an upper bound on $\tilde{d}(v,u)$. 

We can upper bound the length of this path as follows.
\begin{align*}
\tilde{d}(v,u) &\le \sum_{s=1}^{s_j-1}d(v_s, v_{s+1}) + d(v_{s_j},u)\\
&\le \sum_{s=1}^{s_j-1}\frac{\Delta_j}{2^{(s-1)\ell}} + \frac{\Delta_j}{2^{(s_j-1)\ell}} \\
&\le  \frac{\Delta_j}{1-2^{-\ell}}
\end{align*}
\end{proof}

This completes the proof of \Cref{thm:bdd}.
\begin{remark}
By setting $\ell=1$, we get a degree bound of 7 and stretch $8\sqrt{26}$.
\end{remark}

\section{Applications}\label{sec:applications}

\subsection{Compact routing scheme}\label{sec:routing}
Compact routing scheme is a distributed algorithm which routes messages from any given source to any given destination. Given an input metric space $M_X=(X,\|\cdot\|_X)$, the goal is to devise an \emph{overlay network}, \emph{local memory}, and \emph{routing algorithm}. 
The \emph{overlay network} is a subgraph $G=(X, E)$, $E\subseteq \binom{X}{2}$, of $M_X$ on top of which the routing proceeds. Every node $u \in X$ has a unique \emph{port number} assigned to each neighboring edge in the overlay network, ranging between $1$ and $\deg_G(u)$. We focus on the designer port model, where the algorithm designer has a power of assigning the port numbers.
The \emph{local memory} is stored at every node in the network and consists of a \emph{routing table} and a \emph{label (i.e., address)}.  This section is focused on the labeled model, where the algorithm designer can choose the local memory. Given a source node $u$ and a destination node $v$ the \emph{routing algorithm} commences at $u$ and has at its disposal the routing table of $u$ and the label of $v$. Based on this information, it should output the port number of the next node $w$ on the path from $u$ to $v$. The same algorithm is then executed at node $w$, until node $v$ is reached.

There are several basic qualities of a compact routing scheme. The \emph{stretch} is defined as $\max_{u,v\in X} \frac{d_R(u,v)}{\norm{uv}_X}$, where $d_R(u,v)$ is the distance traveled by the routing algorithm when routing from $u$ to $v$. The \emph{routing table size} and the \emph{label size} is the maximum number of bits occupied by the routing table and label, respectively, of any node in $X$. Ideally, a routing scheme will achieve a good tradeoff between the stretch and the size of routing tables and labels while routing on top of a sparse overlay network. For a more detailed introduction to compact routing, refer to \cite{FG01,TZ01,AGMNT08}.

Our routing scheme uses \emph{handshaking} to determine which of the two trees to route on. This is a mechanism by which the source $u$ and the destination $v$ exchange a constant amount of information and agree upon a single bit corresponding to the index of the tree on which the subsequent routing will take place. For more details on handshaking, see \cite{TZ01, ACLRT03, KR25}.

A seminal result by Thorup and Zwick \cite{TZ01} states that every tree admits a routing scheme with stretch 1, no routing tables, and labels of size $(1+o(1))\log{n}$ bits per node. Together with our spanning tree cover, this gives rise to a very simple compact routing scheme --- compute the labels of \cite{TZ01} in both of the trees and for every node in the metric let its label be the concatenation of its labels in the two trees. Clearly, each vertex uses $(2+o(1))\log{n}$ bits for labels and has no routing table. 

The only missing piece is a way to decide which tree to use for routing.
If the coordinates of the input points are at our disposal, we can easily find the right tree to route on. (The details are given in \Cref{sec:local-routing}.) The rest of this section is describing a handshaking protocol which allows us to choose the tree to route on. The result of this section is summarized in the following theorem.

\begin{theorem}
Every set of $n$ points $P \subset \mathbb{R}^2$ admits a compact routing scheme which operates on an overlay network consisting of two trees, has stretch $8\sqrt{26}$, and uses labels of size $(2+o(1))\log{n}$ bits and no routing tables. The routing scheme uses a single round of handshaking in which a constant number of bits is exchanged.
\end{theorem}
\subsubsection{Finding the right tree using handshaking}
For a given source $u$ and destination $v$, we perform the handshaking only in the first tree. A message with a header containing at most 6 bits is passed from $u$ to $v$ using the shortest path in the first tree. Upon reaching $v$, we are able to decide whether the tree has a good stretch or not. This information (consisting of a single bit) is then sent back to $u$.

We first describe a simple version of the handshaking in a Steiner tree. Next, we show how to modify this scheme so that it works in the spanning tree. Finally, we show how to augment the Thorup-Zwick routing scheme with the data necessary for handshaking without asymptotically changing the label sizes.

\paragraph{Handshaking in Steiner trees.} Consider a Steiner tree $T$ from \Cref{thm:steiner}, equipped with the Thorup-Zwick routing scheme. We assume that we can store with every port number in the routing scheme additional 4 bits of information. (We later show how to do this without affecting the bound of $(1+o(1))\log{n}$ bits of memory for labels.)

The 4 bits of information are assigned to port numbers as follows. Consider an arbitrary edge $(u,v)$ in the tree $T$, where $v$ is the parent of $u$. From the construction, we know that $v$ is the center of some recursion square $S$ containing $u$. The additional bits of memory will consist of two parts: \emph{(i)} a single bit set to $0$ for $\port(u,v)$ and set to $1$ for $\port(v,u)$ and \emph{(ii)} a 3-bit string $\pos(u,v)$ describing the relative position of $u$ and $v$, which are assigned as follows.
For $0 \le i \le 7$, let $\bin(i)$ denote the 3-bit binary representation of $i$.
Subdivide $S$ into eight triangles as in \Cref{fig:regions} and set $\pos(u,v)$ to $\bin(0)$ if $u \in A_1(S)$, $\bin(1)$ if $u \in A_2(S)$, ..., and $\bin(7)$ if $u \in D_2(S)$. This concludes the description of assignment of additional information to port numbers. We note that due to the properties of our construction, the edges between the two \emph{internal} vertices in the tree will always have numbers $0$, $2$, $4$, or $6$, because their endpoints are the centers of some square $S$ and one of its four subsquares.

Consider two points $u$ and $v$ and let $S$ be the square such that $u$ and $v$ are in different subsquares of $S$. From the construction, the center $o$ of $S$ is the LCA of $u$ and $v$. By \Cref{obs:choose-tree}, $u$ and $v$ have stretch at most $\sqrt{26}$ in $T$ if they belong to nonconsecutive triangles of $S$. This can be checked using port information as follows. Traverse the path from $u$ to $v$ using the Thorup-Zwick routing scheme. The first bit assigned to every port number will be $0$ while going from $u$ to $o$ and $1$ while going from $o$ to $v$. This allows us to know once we reach the LCA, which is $o$. Let $u=u_0, u_1, u_2, \ldots, u_a=o$ be the vertices on the path from $u$ to $o$.
While traversing this path, we keep three bits of information  describing the relative position of $u$ and $u_i$, denoted by $\pos(u, u_i)$. For $u$ and $u_1$, this information is already encoded in $\port(u,u_1)$. Assuming we know $\pos(u,u_i)$ for some $i\ge 1$, we proceed to describe how to compute $\pos(u,u_{i+1})$ using the information in $\port(u_i,u_{i+1})$. Recall that $\pos(u_i,u_{i+1})$ can only take values $0$, $2$, $4$, or $6$, since $u_i$ and $u_{i+1}$ are internal nodes in the tree for $i \ge 1$. If $\pos(u_i,u_{i+1}) = 0$ and $\pos(u,u_i) \in \{1,2,3,4\}$, then $\pos(u, u_{i+1}) = 1$. We can similarly deduce the other values, as shown in \Cref{tab:otimes}. See also \Cref{fig:otimes} for an explanation. We use the notation $\pos(u,u_{i+1}) = \pos(u,u_i) \otimes \pos(u_i,u_{i+1})$ to denote this computation.
Using this notation, we have the following expressions.
\begin{align*}
\pos(u,o) = \bigotimes_{i=0}^{a-1}\pos(u_i, u_{i+1})
\end{align*}
We can analogously compute $pos(v, o)$. 

\begin{figure}[ht]
    \centering
    \includegraphics[width=0.78\linewidth]{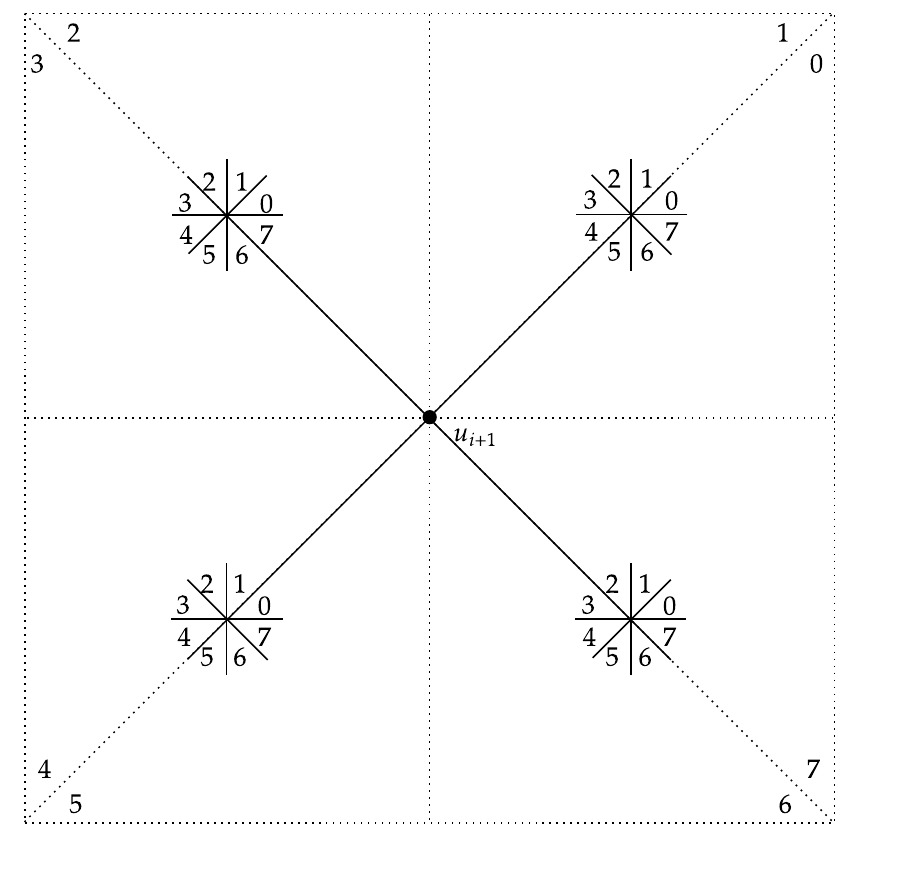}
    \caption{Computing $\pos(u,u_i)\otimes\pos(u_i,u_{i+1})$.}
    \label{fig:otimes}
\end{figure}

\begin{table}[ht]
\centering
\begin{tabular}{cc|cccccccc|}
\cline{3-10}
\multicolumn{1}{l}{}                                      & \multicolumn{1}{l|}{} & \multicolumn{8}{c|}{$\pos(u,u_i)$}                                                                                                                                                       \\ \cline{2-10} 
\multicolumn{1}{c|}{}                                     &                       & \multicolumn{1}{c|}{0} & \multicolumn{1}{c|}{1} & \multicolumn{1}{c|}{2} & \multicolumn{1}{c|}{3} & \multicolumn{1}{c|}{4} & \multicolumn{1}{c|}{5} & \multicolumn{1}{c|}{6} & 7 \\ \hline
\multicolumn{1}{|c|}{\multirow{4}{*}{$\pos(u_i,u_{i+1})$}} & 0                     & \multicolumn{1}{c|}{0} & \multicolumn{1}{c|}{1} & \multicolumn{1}{c|}{1} & \multicolumn{1}{c|}{1} & \multicolumn{1}{c|}{1} & \multicolumn{1}{c|}{0} & \multicolumn{1}{c|}{0} & 0 \\ \cline{2-10} 
\multicolumn{1}{|c|}{}                                    & 2                     & \multicolumn{1}{c|}{2} & \multicolumn{1}{c|}{2} & \multicolumn{1}{c|}{2} & \multicolumn{1}{c|}{3} & \multicolumn{1}{c|}{3} & \multicolumn{1}{c|}{3} & \multicolumn{1}{c|}{3} & 2 \\ \cline{2-10} 
\multicolumn{1}{|c|}{}                                    & 4                     & \multicolumn{1}{c|}{5} & \multicolumn{1}{c|}{4} & \multicolumn{1}{c|}{4} & \multicolumn{1}{c|}{4} & \multicolumn{1}{c|}{4} & \multicolumn{1}{c|}{5} & \multicolumn{1}{c|}{5} & 5 \\ \cline{2-10} 
\multicolumn{1}{|c|}{}                                    & 6                     & \multicolumn{1}{c|}{7} & \multicolumn{1}{c|}{7} & \multicolumn{1}{c|}{7} & \multicolumn{1}{c|}{6} & \multicolumn{1}{c|}{6} & \multicolumn{1}{c|}{6} & \multicolumn{1}{c|}{6} & 7 \\ \hline
\end{tabular}
\caption{Values $\pos(u,u_i) \otimes \pos(u_i,u_{i+1})$.}\label{tab:otimes}
\end{table}

This information tells us exactly whether $u$ and $v$ are in consecutive triangles or not. While going from $u$ to $o$, the header of the routing scheme consists of 3 bits containing $\pos(u,u_i)$. Upon reaching $o$, we compute $\pos(u,o)$ and keep it in the header. Subsequently, we keep 3 bits describing $\pos(o,v_i)$ in the header, together $\pos(u,o)$. Altogether, this requires headers of size $6$ bits. Upon reaching $v$, we can compute $\pos(o,v)$ and decide whether $u$ and $v$ belong to consecutive triangles. We send this single bit of information back to $u$ in the header. This concludes the description of handshaking process for the Steiner trees. We next show how a simple modification allows this process to work with spanning trees.

\paragraph{Handshaking in bounded degree spanning trees.} We next describe how to modify the previous routing scheme so that it works in the spanning tree cover.
We use the tree cover from \Cref{thm:bdd} and focus on a single (e.g., the first) tree $T$. Specifically by \Cref{obs:hierarchy}, every point $u$ participates as a root in at most $\ell+1$ levels in the hierarchy, for some constant $\ell$. Let these levels be $j_0, j_1, \ldots, j_p$ in increasing order. 
For each $1 \le i \le p$, let the corresponding Steiner point be $u^{i}$. For $1 \le i < p$, store locally at $u$ 3 bits describing $\pos(u^{i},u^{i+1})$. Since $p \le \ell+1$, this information occupies at most $3\ell$ bits. Let $(u,v)$ be an arbitrary edge in $T$. Let $j_u$, $j_v$ be such that $u$ is at level $j_u$ and $v$  is at level $j_v$.  Note that $1\le j_u, j_v \le \ell+1$. Let $u^{j_u}$ and $v^{j_v}$ be the corresponding Steiner points. We store with $\port(u,v)$ the following information: $j_u$, $j_v$, $\pos(u^{j_u},v^{j_v})$ and an additional bit denoting whether the edge goes up or down in the tree. This occupies at most $2\log(\ell) + 3$ bits. This concludes the description of the modified information for routing scheme.

We next explain how to find the right tree using the additional information. Upon a routing request from $u$ to $v$, we traverse the path from $u$ to $v$. Let $\ell$ denote the LCA of $u$ and $v$ in $T$. Let $u_0 = u, u_1,\ldots, u_{a-1}, u_{a}=\ell= v_{b}, v_{b-1}, \ldots, v_0 = v$ be the path between $u$ and $v$ in $T$. For $u_0$, let $j_{u_0}^0$ be the lowest level of hierarchy where $u_0$ resides, i.e., the level at which it is the leaf; let $j_{u_0}^1$ be the level at which it connects to $u_1$. For $1 \le i\le a$, let $j_{u_i}^0$ be the level at which the previous edge on the path joins $u_i$ and let $j_{u_i}^1$ be the level at which the next edge leaves $u_i$. For the simplicity of notation, we use $u_i^0 \coloneqq u_i^{j_{u_i}^0}$ to denote the point $u_i$ at level $j_{u_i}^0$. Similarly, we let $u_i^1 \coloneqq u_i^{j_{u_i}^1}$. 
We have  
\begin{align*}
\pos(u, \ell) = 
\bigotimes_{i=0}^{a-1} \left(\pos(u_i^{0}, u_i^{1}) \otimes \pos(u_i^{1}, u_{i+1}^{0}) \right) \otimes \pos(u_a^{0}, u_a^{1})
\end{align*}
We can extract this information as follows. For $u_0$, level $j_{u_0}^0$ is simply the lowest level in the hierarchy. Level $j_{u_0}^1$ is stored in $\port(u_0,u_1)$. Hence, we can compute $\pos(u_0^0,u_0^1)$ using the information locally stored at $u$ and the information in $\port(u_0,u_1)$. We extract $\pos(u_0^1, u_1^0)$ from the corresponding port. For $1 \le i \le a$, indices $j_{u_i}^0$ and $j_{u_i}^1$ are stored in the port numbers of the two edges of the path incident on $u_i$. We can hence compute $\pos(u_i^0, u_i^1)$ using the information stored locally at $u_i$. Information on $\pos(u_i^1, u_{i+1}^0)$ is stored in the corresponding edge. The other piece of information required for finding the right tree is $\pos(v,\ell)$, which can be computed analogously.

Putting everything together, we accumulate information on $\pos(u,\ell)$ on the part of the part of the path from $u$ to $\ell$. Upon reaching $\ell$, we store the corresponding 3 bits in the header. Then, on the path from $\ell$ to $v$, we accumulate information on $\pos(v,\ell)$. Finally, upon reaching $v$, we decide whether $u$ and $v$ are in the consecutive triangles with respect to the square corresponding to $\ell$, and send back to $u$ the one bit of information. This concludes the description of the handshaking process used for deciding the right tree.

\paragraph{Modification of the Thorup-Zwick scheme.} 
The final ingredient of our handshaking process is a modification of the Thorup-Zwick routing scheme \cite{TZ01} which allows us to store a constant number of bits with each port number. 

At a high level, each label of their routing scheme consists of a sequence of variable-length bit strings $s_1,s_2,\ldots,s_k$. To form the label, they store $\code(s_1),\code(s_2),\ldots,\code(s_k)$ so that the appropriate string can be easily extracted. They show (cf. Lemma 2.2 \cite{TZ01}) that $|\code(s)| \le |s| + O(\log(|s|+2))$. In particular, every port number $p$ in their routing scheme is stored using $\code(p)$. Our modification consists of devising a new encoding scheme $\code'(\cdot)$ which satisfies all the requirements of $\code(\cdot)$ and in addition encodes an additional $\alpha$ bits for a constant $\alpha$. Then, we show that using $\code'$ instead of $\code$ does not increase the $(1+o(1))\log{n}$ bound on the label sizes. We first describe the modified encoding scheme.

\begin{lemma}[Modification of the encoding scheme from Lemma 2.2 in \cite{TZ01}]
Let $\alpha$ be an arbitrary constant.
There is an efficient encoding scheme $\code' : \{0,1\}^* \rightarrow \{0,1\}^*$ that satisfies: 
\begin{enumerate}
\item The set $\{ \code'(s) \mid s \in \{0,1\}^* \}$ is an infinite suffix-free code.
\item For every $s \in \{0,1\}^*$ we have $|\code'(s)| \le |s| + \alpha + O(\log(|s|+2))$.
\item In addition, $\code'(s)$ stores $\alpha$ bits of auxiliary information.
\item Given a machine word that contains the string $\code'(s)$ in its least significant bits, it is possible to extract $s$ and $|s|$ and the auxiliary $\alpha$ bits using a constant number of standard operations. It is also possible to remove $\code'$ from this word.
\end{enumerate}
\end{lemma}
\begin{proof}
The proof is similar to the proof of Lemma 2.2 in \cite{TZ01}. Given $i \ge 0$, we use $\bin(i)$ to denote the bit string containing the binary representation of $i$; let $\ell(i) = |\bin(i)|$. For $k \ge \ell(i)$ let $\bin(k,i) = 0^{k-\ell(i)}\bin(i)$, i.e., the binary representation of $i$ padded with leading 0s so that the resulting string has length $k$. We let $m(i) = 2^{\ceil{\log_2{\ell(i)}}}$ for $i \ge 0$ and observe that $\ell(i) \le m(i) < 2\ell(i)$. We use $||s|| = m(|s|)$. Let $a$ be the auxiliary $\alpha$ bits of information to be stored with $s$. Using $.$ to denote a concatenation of binary strings, we define the code of $s$ as follows.
\begin{align*}
\code'(s) = s . a . \bin(||s||, |s|) . \bin(||s||,||s||)
\end{align*}

The length of $\code'(s)$ is $|s| + \alpha + 2m(|s|) \le |s| + \alpha + 4(\log_2|s|+1)$.

We next explain how to extract $s$ from a string $t = s'.\code'(s)$, where $s'$ is some other bit string. First, note that $||s||$ is a power of 2, meaning that it is easily found by looking for the first 1 in $t$. We ignore the first $||s||$ bits of $t$ and take the next $||s||$ bits which represent $|s|$. Next, extract $\alpha$ bits of auxiliary information. Finally, extract $|s|$ bits containing $s$. This concludes the description of decoding.
\end{proof}

We next prove that adding $\alpha$ bits to $\code'(s)$ does not increase the label size beyond $(1+o(1))\log_2(n)$. Specifically, we accordingly modify the proof of Lemma 2.4 in \cite{TZ01}.
\begin{lemma}[Cf. Lemma 2.4 in \cite{TZ01}]
Let $T$ be an $n$-vertex tree. Then, for every vertex $v$ of $T$ we have $|\lbl(v)|\le (1+o(1))\log_2(n)$.
\end{lemma}
\begin{proof}
The identifier $ID(v)$ consists of encodings $3k$ bit stings $c_1,\ldots,c_{3k}$, where $k \le \log_b{n}$ and $b = \ceil{\sqrt{\log_2{n}}}$. In particular, we have $k = O(\log(n)/\log{\log{n}})$. We have $|ID(v)|=\sum_{i=1}^{3k}|\code'(c_i)|= \sum_{i=1}^{3k}|c_i|+3k\alpha+O(\sum_{i=1}^{3k}log(|c_i|+2))$.
In the proof of Lemma 2.4, the contribution of the first and the third term was asymptoticaly bounded by $\log_2(n)+O(\log(n)\cdot \log\log\log(n)/\log\log{n})$. The second term is $3k\alpha = O(\log(n)/\log{\log{n}})$, since $k=O(\log(n)/\log{\log{n}})$ and $\alpha$ is constant. In conclusion, we have $ID(v) = \log_2(n)+O(\log(n)\cdot \log\log\log(n)/\log\log{n})$. The rest of the proof requires no changes.
\end{proof}
\subsection{Path reporting distance oracle}\label{sec:oracle}
In this application, the goal is to preprocess the input point set (in our case $P \subset \mathbb{R}^2$) and construct a compact data structure called a \emph{path reporting distance oracle}. Given a query consisting of two points $a,b \in P$, the data structure should output a \emph{path} of length that approximates the length of $\norm{a-b}$. The goal is to achieve \emph{(i)} fast preprocessing time, \emph{(ii)} data structure of small space, \emph{(iii)} good approximation ratio, and \emph{(iv)} fast query time.

We first construct a Steiner tree cover. By \Cref{lem:time}, both trees in the Steiner tree cover from \Cref{thm:steiner} have $2|P|-1$ vertices and can be computed in $O(|P| \log|P|)$ time. We then run a Steiner point removal procedure, which runs in time $O(|P|)$, on both trees and obtain a spanning tree cover construction with stretch $4\sqrt{26}$ as guaranteed by \Cref{thm:non-steiner}. Specifically, both of the trees have $|P|-1$ vertices.
Using a well-known idea (see, e.g., interval routing scheme \cite{SK85}), we perform a DFS traversal in both of the trees and store at each vertex its DFS timestamp and for each of its children the smallest and the largest DFS timestamp in the corresponding subtree. This takes additional $O(|P|)$ time and space.

Given two query points $a,b \in P$, we first describe how to find the path in a single tree. Given the DFS timestamps it is easy to check whether $a$ is an ancestor of $b$ in constant time. However, finding the right child of $a$ on the path from $a$ to $b$ might incur some additional running time. To avoid this, we always construct the path by going up in the tree. Specifically, we construct $P_a$,  from $a$ to $LCA(a,b)$, and $P_B$, from $b$ to $LCA(a,b)$. Constructing these two paths is straightforward. If $b$ is an ancestor of $a$, add the edge between $a$ and its parent to $P_a$; otherwise, add the edge between $b$ and its parent to $P_b$. Overall, finding the path between $a$ and $b$ in a single tree takes time linear in the length of the path. We employ this path finding process in both trees in parallel. Specifically, in a single step we discover one edge in both of the trees. If after some step the path has been fully discovered in one of the trees, we check whether the stretch of this path is at most $4\sqrt{26}$. If so, we output the path; otherwise, proceed with path finding in the other tree.
In conclusion, we have proved the following theorem.

\begin{theorem}
For every set $P$ in $\mathbb{R}^2$ there exists a path reporting distance oracle $\mathcal{D}$ which can be constructed in $O(|P|\log{|P|})$ time and using $O(|P|)$ space. The basic structure of $\mathcal{D}$ is a union of two trees. For every subsequent query, consisting of points $a,b\in P$, data structure $\mathcal{D}$ outputs a path $\pi_{a,b}$ of stretch $4\sqrt{26}$ in time $O(|\pi_{a,b}|)$, where $|\pi_{a,b}|$ is the number of edges in $\pi_{a,b}$.
\end{theorem}

\subsection{Local routing}\label{sec:local-routing}
This application operates in the online setting. Given a set of points $P \subset \mathbb{R}^2$, the goal is to preprocess the input data and construct an \emph{overlay network} $H = (P,E)$ of small size. An online query from $a \in P$ to $b \in P$ consists of coordinates $a$ and $b$ and  coordinates of neighbors of $a$.
The goal is to decide on the next hop $c$ on the path from $a$ to $b$. The query proceeds recursively until $b$ is reached. The \emph{routing ratio} or \emph{stretch} of the path is the ratio between the length of the traversed path and $\norm{a - b}$. The routing ratio of the algorithm is the maximum routing ratio among all the possible queries. Here, the goal is to achieve \emph{(i)} small overlay network and \emph{(ii)} small routing ratio. 

We explain two local routing results for our tree cover constructions. The first one is using the construction from \Cref{thm:non-steiner}. The trees in this construction do not have a bound on the degree and hence the algorithm might need to access $\Omega(n)$ neighbors of a vertex at some routing step. The second result uses the bounded degree construction and requires additional $\log{n}$ bits of information stored with every coordinate. In both cases, our overlay network consist of the two spanning trees from the corresponding tree cover, and has at most $2|P|-2$ edges.

The first phase in both of the routing algorithms is finding the right tree. Given a query $a,b \in P$, by \Cref{obs:choose-tree}, it suffices to find the red square $S'$ such that $a$ and $b$ are in different subsquares of $S'$. Using the coordinates of $a$ and $b$, a simple binary search can be employed to find $S'$. If $a$ and $b$ belong to non-consecutive triangles of $S'$, we use the red tree; otherwise, we use the blue tree. 

\subsubsection{Local routing using spanning trees}
The overlay network consists of two trees from \Cref{thm:non-steiner}.
We assume that the neighbors of each vertex are ordered so that the parent in the tree is in the first place. For the root of the tree, this place is kept empty. 

Once we found the right tree to route on, the next goal is to decide on the next hop in the path. The following description assumes the right tree is the blue tree $T_B$. (The description is analogous for the red tree.) Use a binary search to find the blue square $S$ where $a$ and $b$ belong to different subsquares. Using Item 1 in \Cref{obs:navigation}, we can decide if $a$ is the root of the subtree of $T_B$ corresponding to $S$. If $a$ is not the root, the next hop is to the parent of $a$. If $a$ is the root, we choose the unique edge towards the subsquare of $S$ where $b$ resides. The uniqueness is guaranteed by Item 2 in \Cref{obs:navigation}.
This completes the description of the algorithm for finding the next edge in the tree. In conclusion, we have proved the following theorem.
 
\begin{theorem}
For every set of points $P \subseteq \mathbb{R}^2$, there is a memoryless local routing algorithm operating on an overlay network of size $2|P|-2$ (consisting of two trees) achieving the routing ratio of $4\sqrt{26}$.
\end{theorem}

\subsubsection{Local routing using bounded degree spanning trees}
The overlay network consists of two trees from \Cref{thm:bdd}.
In each of the two trees, equip every vertex with a DFS index (occupying $\log|P|$ bits) and use the interval routing scheme \cite{SK85} for navigation. 
\begin{theorem}
For every set of points $P \subseteq \mathbb{R}^2$, there is a memoryless local routing algorithm operating on an overlay network of size $2|P|-2$ (consisting of two trees), where each tree has degree at most 7, achieving the routing ratio of $8\sqrt{26}$. Every point in the overlay network is equipped with $2\log{|P|}$ bits of information stored with its coordinates.
\end{theorem}
\bibliographystyle{alpha}
\bibliography{refs,shay}
\appendix
\section{Number of edges in Yao and \texorpdfstring{\bm{$\Theta$}}{Θ} graphs}\label{sec:cones}
In this section, we give an example where Yao$_4$ and $\Theta_4$ graphs for a set of $n$ points in $\mathbb{R}^2$ have at least $3n-7$ edges.
\begin{definition}\label{def:theta_yao}
Let $P$ be a set of points in $\mathbb{R}^2$ and let $k \ge 2$ be an arbitrary integer. For each point $p$ in $P$, and for each $0 \le i < k$, let $\mathcal{C}_i(p)$ be the cone consisting of the points between rays with angles $2i\pi/k$ and $2(i+1)\pi/k$ emanating from $p$. Yao$_k$ graph is obtained by connecting every point $p \in P$ with the closest point in each of $\mathcal{C}_i(p)$, for $0 \le i < k$. $\Theta_k$ graph is obtained by connecting every point $p \in P$ with the point having the closest projection onto the bisector of $\mathcal{C}_i(p)$, for $0 \le i < k$.
\end{definition}
From the construction it is clear that both Yao$_4$ and $\Theta_4$ graphs have at most $4n$ edges. We next describe an instance where these graphs have at least $3n-7$ edges.
\begin{figure}[H]
    \centering
    \includegraphics[width=0.9\linewidth]{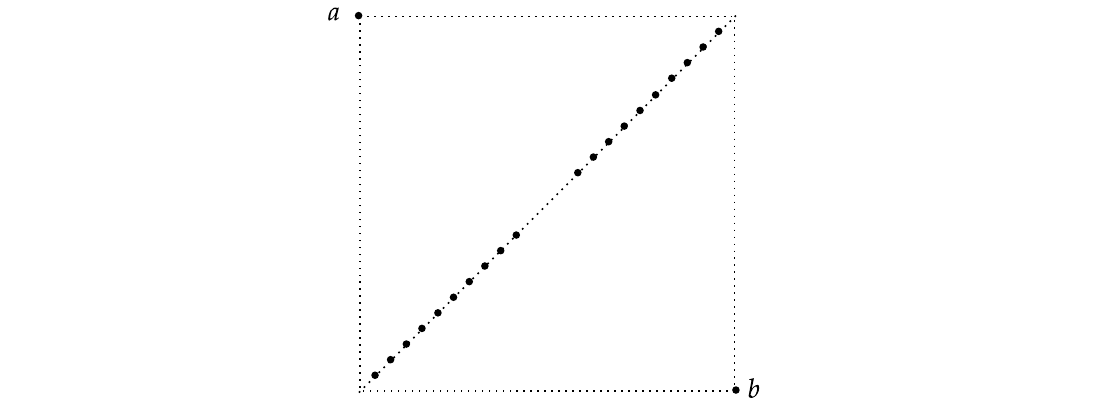}
    \caption{An instance of a point set for which both Yao$_4$ and $\Theta_4$ graphs have $3n-7$ edges.}
    \label{fig:theta}
\end{figure}
Consider a set of points such that points $a$ and $b$ are in the antipodal corners of a square and the remaining $n-2$ points lie on a diagonal, as shown in \Cref{fig:theta}. For every point $p$ on the diagonal, $\mathcal{C}_1(p)$ contains only point $a$ and $\mathcal{C}_3(p)$ contains only point $b$. This contributes $2n-4$ edges to both Yao$_4$ and $\Theta_4$ graphs. There is also is an edge between every two consecutive points on the line, contributing to $n-3$ edges. Overall, the number of edges is $3n-7$.

\section{Tree cover using the shifting lemma}\label{sec:shifting}
In this section, we show that the following lemma from \cite{Cha98} implies a Steiner tree cover with 3 trees and stretch $6\sqrt{2}$ for a set of points in $\mathbb{R}^d$.

\begin{lemma}[Shifting lemma \cite{Cha98}]
Consider any two points $p, q \in [0,1)^d$, and let $\mathcal{T}$ be the infinite quadtree of $[0, 2)^d$. For $D = 2 \ceil{d/2}$ and $i = 0,\ldots, D$, let $v_i = (i/(D + 1),\ldots , i/(D + 1))$. Then there exists an $i \in \{0, \ldots , D\}$, such that $p + v_i$ and $q + v_i$ are contained in a cell of $\mathcal{T}$ with side length at most $2(D + 1) \norm{p-q}$.
\end{lemma}

For a set of points $P \subset [0,1)^2$, parameter $D$ in the lemma takes value $D=2$. For each of the shifts $v_i$, $0\le i\le 2$, the corresponding tree $T_i$ is obtained on top of $P$ shifted by $v_i$ as using the following recursive procedure. Root the $T_i$ in the center $r$ of the square $S=[0,2) \times [0,2)$ and connect it to the roots of the trees recursively obtained in each of the four subsquares of $S$. The recursion stops whenever currently considered square contains only one point, in which case we connect the center of the square with that point. This concludes the description of the recursive procedure.

To analyze the stretch, observe first that two points $p$ and $q$ in some quadtree cell with side length $\Delta$ have tree distance at most $d_T(p,q) \le \Delta \sqrt{2}$. Using the lemma, we consider the shift $v_i$ such that $p+v_i$ and $q+v_i$ belong to some cell with side length at most $2(D+1)\norm{p-q} = 6\norm{p-q}$. We have $d_T(p,q) \le \Delta\sqrt{2} \le 6\sqrt{2}\norm{p-q}$.

\end{document}